\documentclass{scrartcl}

\usepackage{amssymb}
\usepackage{amsmath}
\usepackage{amsthm}
\usepackage{mathrsfs}
\usepackage{hyperref}
\usepackage{tikz}
\usepackage{pgffor}

\hypersetup{colorlinks=true,citecolor=blue!70!black,linkcolor=red!85!black,urlcolor=blue}

\def\defas {:=}
\theoremstyle{plain}
\newtheorem{thm}{Theorem}[section]
\newtheorem{prop}[thm]{Proposition}
\newtheorem{lem}[thm]{Lemma}
\newtheorem*{cor}{Corollary}
\theoremstyle{remark}
\newtheorem{rem}[thm]{Remark}
\theoremstyle{definition}
\newtheorem{defn}[thm]{Definition}

\newcommand{\void}[1]{}

\def\N {\mathbb{N}}
\def\Z {\mathbb{Z}}

\def\R {\mathbb{R}}
\def\C {\mathbb{C}}

\def\id {\mathrm{id}}

\def\oti {\ensuremath{\otimes}}
\def\woti{\ensuremath{\otimes^\mathrm{H}}}
\def\wop{\ensuremath{\oplus^\mathrm{H}}}
\newcommand{\fuse}{\ensuremath{\widehat{\otimes}}}
\newcommand{\cont}{\ensuremath{\curlyvee}}
\newcommand{\env}[1]{\ensuremath{\mathcal{U}(\mathfrak{#1})}}
\newcommand{\ens}{\ensuremath{\mathcal{U}(\mathfrak{sl}_2)}}
\newcommand{\ena}{\ensuremath{\mathcal{U}(\widehat{\mathfrak{sl}}_2)}}
\newcommand{\uD}{\ensuremath{{}^\mathrm{H}\! D}}

\def\Hom {\mathrm{Hom}}
\def\oti {\otimes}

\def\End {\mathrm{End}}
\def\2cat {\mathcal{K}}
\def\PoincDual {\xymatrix{&\ar@{|->}[l] \ar@{|->}[r]&}}

\def\one {{1\!\!1}}

\numberwithin{equation}{section}

\begin{document}

\title{On duality and extended chiral symmetry in the $SL(2,\R)$ WZW model}

\author{Jens Fjelstad}

\date{}

\maketitle

\thispagestyle{myheadings}

 \begin{center}
Department of Physics, Nanjing University\\
 22 Hankou Road, Nanjing, 210093 China\\
 \href{mailto:jens.fjelstad@gmail.com}{\tt jens.fjelstad@gmail.com}
 \end{center}

\begin{abstract}
Two chiral aspects of the $SL(2,\R)$ WZW model in an operator formalism are investigated. First, the meaning of duality, or conjugation, of primary fields is clarified. On a class of modules obtained from the discrete series it is shown, by looking at spaces of two--point conformal blocks, that a natural definition of contragredient module provides a suitable notion of conjugation of primary fields, consistent with known two--point functions. We find strong indications that an apparent contradiction with the Clebsch--Gordan series of $SL(2,\R)$, and proposed fusion rules, is explained by nonsemisimplicity of a certain category. Second, results indicating an infinite cyclic simple current group, corresponding to spectral flow automorphisms, are presented. In particular, the subgroup corresponding to even spectral flow provides part of a hypothetical extended chiral algebra resulting in proposed modular invariant bulk spectra.
\end{abstract}

\tableofcontents

%%%%%%%%%%%%%%%%%%%%%%%%%%%%%%%%%%%%%%%%%%%%%%%%%%%%%%%%%
\section{Introduction}\label{sec:intr}

Since the seminal work of Maldacena and Ooguri on the $AdS_3$ string \cite{MaOg}, the $SL(2,\R)$ WZW model is more often than not identified with the WZW model on the universal cover of $SL(2,\R)$. Correlators in this model were subsequently \cite{MaOg2} defined by analytic continuation of correlators in the so called $H^+_3$ (or $SL(2,\C)/SU(2)$) WZW model. This is of course not for lack of motivation; for target space physics it is necessary to get rid of closed timelike curves, so turning to the universal cover of $SL(2,\R)$ is a natural choice. Since the Killing metric on $SL(2,\R)$ is of Lorentzian signature one expects the need of regularization in analogy with target space Wick rotation for the string in flat spacetime, thus also the use of the $H^+_3$ model is natural. In contrast, we will here mainly deal with the WZW model on the actual $SL(2,\R)$ group manifold, as reflected in the chosen set of representations, and on world sheets with conformal structures of Euclidean signature. Furthermore, the focus is on questions related to this model as a (non unitary) CFT as opposed to as a string model.
More precisely, the aim of this article is to clarify some questions concerning how structures familiar from rational CFT will have to generalize, or not, to fit the $SL(2,\R)$ WZW model.  Underlying a rational CFT is category of chiral objects (chiral category for short) with the structure of a modular tensor category.
We aim here to clarify parts of the structure of the corresponding chiral category in the $SL(2,\R)$ WZW model.\\

One aspect that will be investigated concerns the nature of duality in the chiral category. In the rational case chiral objects have duals, and dual simple objects correspond to conjugate primary fields. In the $SL(2,\R)$ model it is usually assumed that the spectrum corresponds to a subset of the discrete and principal continuous series of unitary representations of $sl(2,\R)$. From the perspective of an operator formalism, one puzzling aspect of the $SL(2,\R)$ model is that the Clebsch--Gordan series of these representations do not contain the one--dimensional  trivial representation. In contrast to a rational WZW model it thus seems that the representation theory of the horizontal subalgebra itself does not have a duality that can be reflected in a duality of representations of a corresponding affine Lie algebra. Furthermore, in the infinite cover case the fusion rules proposed in \cite{Nun2} correspond to the Clebsch--Gordan series of $sl(2,\R)$ in that the identity field is not present. On the other hand, one can easily identify conjugate primary fields from the two--point functions in \cite{MaOg2}. This begs the question of precisely what structure in the chiral category corresponds to conjugation of primary fields.

To this end we show that by restricting to weight representations contained in the unitary discrete series of representations of $sl(2,\R)$ one actually gets intertwiners from the tensor product of a simple module with its contragredient to the trivial module. Furthermore, there are no nonzero intertwiners in the other direction, implying that the tensor product is actually reducible but indecomposable. This result is presented as Proposition \ref{prop:sl2nonsemi}. In order to verify that this structure is lifted to the affine Lie algebra we then turn to an analysis of the $\ena$--modules obtained by prolonging (discrete series) weight modules and the trivial module of $\ens$, and by acting on these with spectral flow automorphisms. After deriving various results concerning the structure of these modules in section \ref{sec:proldiscr}, we give a definition of contragredient of these modules. It is then shown in Theorem \ref{thm:cont} that this class of modules is closed under the operation of taking contragredients, and that the contragredient reflects that of the horizontal submodules as well as the conjugate of primary fields. The main result of section \ref{sec:mods} is Theorem \ref{thm:2ptblocks}, showing that the space of two--point conformal blocks on the sphere is of dimension one iff the two points are labeled by a module and its contragredient, and zero--dimensional otherwise. Under the assumption that spaces of conformal blocks on the sphere are, as usual, hom spaces in the chiral category from the fusion product of the modules labeling the marked points to the vacuum module (corresponding to the identity field), this result shows that at least half of the structure of Proposition \ref{prop:sl2nonsemi} lifts to the affine case. The other half, nonexistence of morphisms in the other direction, cannot be investigated with current knowledge of the $SL(2,\R)$ model. However, considering the relation between the tensor product of horizontal submodules and the fusion product in the rational case, it seems almost inevitable that also the other direction will lift to the affine case. The picture that appears of the chiral category in the $SL(2,\R)$ WZW model is that of a nonsemisimple tensor category, where the fusion product of two simple objects is not necessarily decomposable as a direct sum of simple objects. Furthermore, the asymmetry in Proposition \ref{prop:sl2nonsemi}, when lifted to the affine case implies that the chiral category will not have a duality (i.e., it will not have a rigid structure). There may still be a structure weaker than rigidity, however, and we note that the results presented here are still consistent with a so called semi--rigid structure \cite{Miyamoto}.\\

The second aspect that will be investigated is related to extensions of the chiral symmetry. Recall from the rational case that there is a class of extensions of chiral algebras known as simple current extensions. Simple currents are primary fields that correspond to invertible objects in the chiral category. The authors of  \cite{HeHwRS} proposed modular invariant bulk spectra for the $SL(2,\R)$ WZW model on the single cover. The structure of these indicate an extended chiral algebra, partly related to even spectral flow automorphisms. Motivated by an analogy with rational WZW models it is natural to expect the chiral objects obtained by acting with spectral flow on the vacuum module to be invertible, i.e. correspond to simple currents. We show in Theorem \ref{thm:unit} that the vacuum module possesses properties required of the tensor unit of the chiral category. In Theorem \ref{thm:sfconfblocks} it is then shown that modules obtained from spectral flow on the vacuum module satisfy properties required of invertible objects inside certain hom spaces of the chiral category. We only manage to show this when all but one of the other objects in the hom space are prolongation modules of \ena, prolonged from \ens--modules, but believe that going beyond this restriction is merely a question of technical difficulty. The results of section \ref{sec:extchir} are interpreted as strong indications that the modules obtained by acting with spectral flow on the vacuum module are indeed invertible, and correspond to simple currents. As a result we can identify an integer spin simple current extension partly responsible for the bulk spectra presented in \cite{HeHwRS}.\\

The paper is organized as follows.

Section \ref{sec:structures} is devoted to a review of some relevant aspects of rational CFT, and a discussion concerning which aspects might survive in the $SL(2,\R)$ model. In particular we identify a class of \ena--modules, related to the discrete series and the trivial representation of $sl(2,\R)$, expected to be relevant. These modules form the objects of study in this paper.

Section \ref{sec:mods} contains an investigation of the properties of this class of modules.
 First some basic facts concerning the zero mode representations are reviewed, in particular aspects related to unitarity vs. unitarizability. We show that, while the (Hilbert space) tensor product of \emph{unitary} representations $ \uD^+_j$ and $ \uD^-_j$ does not contain the trivial representation, there is an intertwiner from the (vector space) tensor product of the underlying weight modules, $D^+_j$ and $D^-_j$, to the trivial representation. We therefore conjecture that the chiral category $\mathcal{C}^k$ contains $\mathcal{U}(\widehat{\mathfrak{sl}}_2)$--modules constructed, not from the \emph{unitary} discrete series but from the underlying linear (and unitarizable) discrete series of weight modules. Next, some of these $\mathcal{U}(\mathfrak{sl}_2)$--modules are prolonged to $\mathcal{U}(\widehat{\mathfrak{sl}}_2)$--modules, including also new modules obtained by spectral flow. A definition of the contragredient of these modules is proposed, and it is shown that the modules prolonged from $D^-_j$ are contragredient to those prolonged from $D^+_j$, together with analogous results for modules obtained by spectral flow (also including modules generated from the vacuum module). By determining suitable spaces of two--point conformal blocks it is shown that contragredient modules also correspond to conjugate primary fields.

In section \ref{sec:extchir} results concerning conformal blocks on $\C\mathbb{P}^1$ are established which strongly indicate that the modules obtained by acting with spectral flow on the vacuum module $D^k_0$ are invertible. The orbits of modules $D^{\pm,k}_j$ under the subgroup of the simple current group given by even elements are bricks $X^{\pm,k}_j$ occurring in previously proposed bulk spectra, and we thus identify part of the extended chiral algebra as a (integer spin) simple current extension.

The final section contains brief discussions of some of the results, and of possible future directions.

%%%%%%%%%%%%%%%%%%%%%%%%%%%%%%%%%%%%%%%%%%%%%%%%%%%%%%%%%
\section{Structures in rational CFT and the \texorpdfstring{$SL(2,\R)$}{SL(2,R)} WZW model}\label{sec:structures}
Since we aim to compare structures appearing in the $SL(2,\R)$ WZW model with structures familiar from rational CFT it is suitable to first review some aspects of the rational case, and then discuss more generally what can be expected in the $SL(2,\R)$ case.

\subsection{Aspects of rational CFT}\label{sec:rcft}
Let us first briefly review (some of) the relevant structures appearing in rational CFT. The symmetries of rational models are encoded in \emph{chiral algebras}, which on $\C\mathbb{P}^1$ can be described as rational conformal vertex algebras. Fix such a vertex algebra $\mathfrak{V}$. One often splits the structure (and indeed, the construction) of rational CFT into \emph{chiral CFT} and the remaining part of finding the physical correlators among \emph{chiral correlators}.\\

Essential for chiral CFT is the category $\mathrm{Rep}(\mathfrak{V})$ of modules of $\mathfrak{V}$. It is known \cite{Huang} under certain requirements (which we assume to be satisfied) to possess the structure of a modular tensor category, in particular implying that it is an abelian, $\C$--linear, finite, semisimple ribbon category. Semisimplicity means that every $\mathfrak{V}$--module is equivalent to the direct sum of finitely many simple modules. In particular, $\mathfrak{V}$--modules are completely reducible. Finite refers to the set of isomorphism classes of simple modules being finite. Denote by $\{S_i\}_{i\in \mathcal{I}}$ a set of representatives of these isomorphism classes. The property ribbon implies a number of structures. 

First, it is a tensor category; there exists a tensor product, here denoted $\fuse$, of $\mathfrak{V}$--modules, with a simple tensor unit $\one$, the vacuum module. Let $S_0$ be the chosen representative isomorphic to the tensor unit. In a rational CFT, $\fuse$ is the fusion product and $\one$ corresponds to the identity field, always assumed to be present. The tensor product induces a product on the set of isomorphism classes of modules, thus providing $K_0(\mathrm{Rep}(\mathfrak{V}))$ with the structure of a ring with a distinguished basis. The structure constants of this ring is precisely the fusion rules of the chiral CFT.

Second, it is rigid; there exist left and right dualities, coinciding on objects. A (right) duality associates to every object $X$ a (right) dual $X^\vee$ together with distinguished morphisms $b_X\in\Hom(\one,X\fuse X^\vee)$ and $d_X\in\Hom(X^\vee\fuse X,\one)$, satisfying
\[(\id_X\fuse d_X)\circ(b_X\fuse\id_X)=\id_X,\ (d_X\fuse\id_{X^\vee})\circ(\id_{X^\vee}\fuse b_X) = \id_{X^\vee}.\]
The distinguished morphisms of a left duality differs in the order of $X$ and $X^\vee$.
By $S_{\bar\iota}\defas S_i^\vee$, the duality defines an involution of $\mathcal{I}$, $\mathcal{I}\ni i\mapsto \bar{\iota}\in\mathcal{I}$, fixing the distinguished element $0$.

Third, $\mathrm{Rep}(\mathfrak{V})$ is braided; for any two $\mathfrak{V}$--modules $X$ and $Y$ there exists an isomorphism $c_{X,Y}:X\fuse Y\stackrel{\sim}{\rightarrow} Y\fuse X$ satisfying Yang--Baxter relations.

Finally, it is balanced; every $\mathfrak{V}$--module $U$ comes with a distinguished automorphism $\Theta_U$ (the twist), satisfying certain relations expressing compatibility with the tensor product, duality, and braiding. On simple $\mathfrak{V}$--modules the twist takes the form $\Theta_{S_i}=\theta_i\mathrm{id}_{S_i}$ for a root of unity $\theta_i=e^{-2\pi i\Delta_i}$, where $\Delta_i$ is the conformal weight of the primary field corresponding to $S_i$.

The prefix ``modular'' refers to the additional requirement of invertibility of a $|\mathcal{I}|\times |\mathcal{I}|$ matrix $S$ constructed with the help of the braiding and duality in $\mathrm{Rep}(\mathfrak{V})$. This modular $S$ matrix represents the modular transformation $\tau\mapsto -1/\tau$ on the $|\mathcal{I}|$--dimensional vector space spanned by characters of $\mathfrak{V}$--modules. One can also represent the transformation $\tau\mapsto \tau +1$ with an invertible matrix, leading to a representation of the modular group $SL(2,\Z)$.

With the help of the category $\mathrm{Rep}(\mathfrak{V})$ one constructs a modular functor, assigning (finite dimensional) vector spaces of conformal blocks to closed Riemann surfaces with certain additional structure, in particular marked points labeled by objects of $\mathrm{Rep}(\mathfrak{V})$. These vector spaces combine to vector bundles $\mathfrak{B}$ with projectively flat connections over the moduli spaces of complex curves, and the conformal blocks (or chiral correlators) are \emph{holomorphic sections} of these bundles.  The dual (or ``conjugate'') $\Phi_{\bar{\iota}}$ of a primary field $\Phi_i$ corresponding to $S_i$ in a rational CFT then satisfies the defining property of being the unique primary field $\Psi$ such that the conformal block on the sphere with insertions of $\Phi_i$ and $\Psi$ is nonvanishing.\\

Spaces of conformal blocks will play an important role further down, so let us recall the definition (see \cite{FB} for the vertex algebra construction, or \cite{BK} for the case of an affine Lie algebra). Restrict to the case of the vertex algebra $\mathfrak{V}_{\mathfrak{g},k}$ related to an untwisted affine Lie algebra $\widehat{\mathfrak{g}}$ at level $k\in\N$. Simple objects correspond to simple integrable highest weight modules, $L^k_\lambda$, of $\mathcal{U}(\widehat{\mathfrak{g}})$. Denote by $V^k_\lambda$ the Verma module with the same highest weight $\lambda$ as $L^k_\lambda$. 	Let $C$ be a smooth complex curve, $p_1,\ldots,p_n\in C$ distinct marked points on $C$ with choice of local coordinates $z_1,\ldots,z_n$ in neighbourhoods of the marked points, and $\lambda_1,\ldots,\lambda_n$ integral dominant weights of $\widehat{\mathfrak{g}}$. If $\mathfrak{g}$ denotes the simple Lie algebra isomorphic to the horizontal subalgebra of $\widehat{\mathfrak{g}}$, define
\[\mathfrak{g}(C-\vec{p})\defas\mathfrak{g}\oti \mathcal{O}_{C\backslash\{\vec{p}\}},\]
the Lie algebra of $\mathfrak{g}$--valued functions regular outside the marked points $p_1,\ldots,p_n$, and with at worst finite order poles at the marked points. The algebra $\mathfrak{g}(C-\vec{p})$ acts on $V^k_{\vec{\lambda}}=V^k_{\lambda_1}\oti V^k_{\lambda_2}\oti\cdots \oti V^k_{\lambda_n}$ and $L^k_{\vec{\lambda}}$ as follows. Use the loop algebra construction of $\widehat{\mathfrak{g}}$ to identify the Laurent expansion of $X\oti f$ in $z_i$ with an element $X_{f_i}\in \widehat{\mathfrak{g}}$. The action of $X\oti f\in\mathfrak{g}(C-\vec{p})$ is then defined as
\[X_{f_1}\oti \mathbf{1}+\mathbf{1}\oti X_{f_2}\oti \mathbf{1} + \ldots +\mathbf{1}\oti X_{f_n}.\]
That this indeed defines a representation of $\mathfrak{g}(C-\vec{p})$ is ensured by the residue theorem.
\begin{defn}\label{def:confblock}
	The vector space of conformal blocks, $\mathcal{H}(C,\vec{p},L^k_{\vec{\lambda}})$, associated to $C$, $\vec{p}$, and $L^k_{\vec{\lambda}}$, is defined as the space of coinvariants of the $\mathfrak{g}(C-\vec{p})$ action:
	\begin{equation}
		\mathcal{H}(C,\vec{p},L^k_{\vec{\lambda}})\defas \Big(L^k_{\lambda_1}\oti\cdots\oti L^k_{\lambda_n}\Big)^*_{\mathfrak{g}(C-\vec{p})},
	\end{equation}
	i.e. it consists of functionals $\beta$ satisfying
	\[\beta\circ\left(X_{f_1}+X_{f_2}+\ldots+X_{f_n}\right)=0\]
	for all $X\oti f\in\mathfrak{g}(C-\vec{p})$. Factors of $\mathbf{1}$ have been neglected to improve readability.
\end{defn}
Spaces of conformal blocks satisfy a number of important properties, some of which are summarized in the following
\begin{prop}\label{prop:confblockprops}~
\begin{enumerate}
	\item[(i)] For any smooth compact curve $C$, any finite nonempty set of marked points $\vec{p}$ on $C$, and any integrable $L_{\vec{\lambda}}^k$, the vector space $\mathcal{H}(C,\vec{p},L^k_{\vec{\lambda}})$ is finite dimensional.
	\item[(ii)] $\mathcal{H}(C,\vec{p},L^k_{\vec{\lambda}})\cong\mathcal{H}(C,\vec{p},V^k_{\vec{\lambda}})$
	\item[(iii)] $\mathcal{H}(C,\vec{p},L^k_{\vec{\lambda}})$ is, up to isomorphism, independent of the choice of local coordinates. This a'posteriori justifies the abuse of notation of leaving out $\vec{z}$ in the notation.
	\item[(iv)] $\mathcal{H}(C,\vec{p},L^k_{\vec{\lambda}})\cong\mathcal{H}(C,\vec{q},L^k_{\vec{\lambda}})$, where a distinguished isomorphism is given by parallel transport with respect to a projectively flat connection.
	\item[(v)] For $C=\C\mathbb{P}^1$, use the notation $\mathcal{H}(\vec{p},L^k_{\vec{\lambda}})\defas\mathcal{H}(\C\mathbb{P}^1,\vec{p},L^k_{\vec{\lambda}})$. Then there is the following linear isomorphism
	$\mathcal{H}(\vec{p},L^k_{\vec{\lambda}})\cong\Hom_{\mathrm{Rep}(\mathfrak{V}_{\mathfrak{g},k})}(L^k_{\lambda_1}\fuse L^k_{\lambda_2}\fuse\cdots\fuse L^k_{\lambda_n},\one)$.
\end{enumerate}
\end{prop}
Proofs of these properties can be found in, for instance, \cite{TUY, BK} with the minor difference that they are dealing with the duals of the vector spaces $\mathcal{H}(C,\vec{p},L^k_{\vec{\lambda}})$. Note that one consequence of this difference is that property $(v)$ is in \cite{BK} modified by replacing the right hand side with $\Hom_{\mathrm{Rep}(\mathfrak{V}_{\mathfrak{g},k})}(\one,L^k_{\lambda_1}\fuse L^k_{\lambda_2}\fuse\cdots\fuse L^k_{\lambda_n})$. Since the category $\mathrm{Rep}(\mathfrak{V}_{\mathfrak{g},k})$ is rigid and semisimple these two hom spaces are canonically dual, and due to property $(i)$ they are even isomorphic. For nonrigid or nonsemisimple categories, however, it matters greatly which of the two hom spaces appears on the right hand side of $(v)$.\\

We will be very brief concerning the second part, constructing physical correlation functions, and discuss mainly what will be of direct relevance for this article. The principle of holomorphic factorization implies that physical correlators on a worldsheet $\Sigma$ are found among chiral correlators on the complex double, $\widehat{\Sigma}$, of $\Sigma$. It has been shown \cite{FRS1,FjFRS1} that the task of finding physical correlation functions can then be reformulated as that of finding a so called symmetric special Frobenius algebra $A$ in $\mathrm{Rep}(\mathfrak{V})$.

There is in any modular category a symmetric special Frobenius algebra on the object $\one$. The corresponding full CFT is sometimes called the Cardy theory, and always gives the charge conjugation torus partition function. In addition there is a class of Frobenius algebras, the so called Schellekens algebras \cite{FRS3}, that can be canonically constructed. These algebras are responsible for (all) torus partition functions of simple current type \cite{SchYank}. Recall that a simple current in a rational CFT is a primary field $J$ with the property that there exists another primary field $J^{-1}$ (which then must be the conjugate of $J$) such that $J\fuse J^{-1}\cong\one\cong J^{-1}\fuse J$. In the chiral category $\mathcal{C}=\mathrm{Rep}(\mathfrak{V})$ a simple current is an invertible object, and the set of invertible objects gives a multiplicative subgroup of the fusion ring $K_0(\mathcal{C})$, the Picard group $\mathrm{Pic}(\mathcal{C})$. The effective center, $\mathrm{Pic}^\circ(\mathcal{C})\subset\mathrm{Pic}(\mathcal{C})$, is the subset of the Picard group of $\mathcal{C}$ consisting of elements $[J]$ whose order $N_J$ satisfies $N_J\Delta_J\in\Z$, where $\Delta_J$ is the conformal weight of the simple current $J$. To every sub\emph{group} $H\subset\mathrm{Pic}^\circ(\mathcal{C})$ one can associate at least one (sometimes more than one) Schellekens algebra such that the isomorphism classes of its simple subobjects form precisely $H$.
Mutually local integer spin simple currents have particularly nice properties. For instance, one can then identify an extension of the chiral algebra (as an object in $\mathcal{C}$) as\begin{equation}
	\bigoplus_{i\in I} S_j^{\oplus Z_{i,0}},
	\label{eq:extchiralg}
\end{equation}
where $Z_{i,j}$ are the coefficients of the torus partition function. The bulk state space of the CFT then takes the form
\begin{equation}
	\mathscr{H} = \bigoplus_{\alpha,\beta} \left(X_\alpha \oti_\C X_\beta\right)^{\oplus Z'_{\alpha,\beta}},
	\label{eq:scbulkspace}
\end{equation}
where each chiral brick $X_\alpha$ consists of an $H$ orbit of a simple object, $X_\alpha = \oplus_{j\in H i_{\alpha}} S_{j}$ for some $i_\alpha\in I$. The modular invariant $Z_{i,j}$ is of course immediately obtained from $Z'_{\alpha,\beta}$. The chiral bricks $X_\alpha$ correspond to simple objects in the category of representations of the extended chiral algebra.

%%%%%%%%%%%%%%%%%%%%%%%%%%%%%%%%%%%%%%%%%%%%%%%%%%%%%%%%%
\subsection{Expected structures in the \texorpdfstring{$SL(2,\R)$}{SL(2,R)} WZW model}\label{sec:nrcft}
It is of interest how the structures reviewed above for rational CFT generalize, or not as may happen, to nonrational CFT in general and the $SL(2,\R)$ WZW model in particular. The notion of a vertex algebra is essentially a formalization of the operator formalism of CFT, so if an operator formalism exists for the $SL(2,\R)$ WZW model it is more or less automatic that there is also a notion of conformal vertex algebra covering this model (see remarks below in section \ref{sec:proldiscr}). Assume henceforth that  there exists a corresponding type of vertex algebra with representation category $\mathcal{C}$. Refer for definiteness to $\mathcal{C}$ as the chiral category, and the objects thereof as chiral objects.\\

At present the nature of this hypothetical vertex algebra is unclear. It is, however, natural to expect that the chiral objects $X\in\mathrm{Ob}(\mathcal{C})$ are in particular modules of the centrally extended loop algebra of $sl(2,\R)$, and it is henceforth in the capacity of such modules they will be discussed. Take as $sl(2,\R)$ the complex algebra $\mathfrak{sl}_2$ generated by $J^\pm$, $J^0$ satisfying
\begin{align}
	~[J^+,J^-] &= J^0\nonumber\\
	~[J^0,J^\pm] &= \pm 2J^\pm,\label{sl2}
\end{align}
together with the antiinvolution $\widetilde{\omega}:J^0\mapsto J^0$, $J^\pm\mapsto -J^\mp$. 
Let us also identify the centrally extended loop algebra with the pair $(\widehat{\mathfrak{sl}}_2,\omega)$, where the untwisted affine Lie algebra $\widehat{\mathfrak{sl}}_2$ is generated by $J^\pm_n$, $J^0_n$, $K$, $d$, $n\in\Z$ satisfying
\begin{align}
	~[J^+_m,J^-_n] &= J^0_{m+n}+K m\delta_{m+n}\nonumber\\
	~[J^0_m,J^\pm_n] &= \pm 2J^\pm_{m+n}\nonumber\\
	~[J^0_m,J^0_n] &= 2Km\delta_{m+n}\nonumber\\
	~[d,J^a_n] &= nJ^a_n\nonumber\\
	~[K,J^a_n] & = 0,\label{affsl2}
\end{align}
and the antiinvolution $\omega$ is defined as $\omega:$ $J^0_n\mapsto J^0_{-n}$, $J^\pm_n\mapsto -J^\mp_{-n}$, $K\mapsto K$, $d\mapsto d$. The horizontal subalgebra generated by $J^\pm_0$, $J^0_0$, together with the restriction of $\omega$, will be identified with $sl(2,\R)$.

Consider the following decomposition of $\widehat{\mathfrak{sl}}_2$:
\begin{equation*}
	\widehat{\mathfrak{sl}}_2 = \mathfrak{r}_-\oplus\C J^-_0\oplus\mathfrak{h}\oplus\C J^+_0\oplus\mathfrak{r}_+,
\end{equation*}
where $\mathfrak{r}_\pm = \mathrm{span}_\C\{ J^+_n, J^-_n,J^0_n\}_{n\in\pm\N}$ and $\mathfrak{h}=\mathrm{span}_\C\{J^0_0,K,d\}$. The definitions $\mathfrak{n}_\pm = \C J^\pm_0\oplus\mathfrak{r}_\pm$ then yield a Cartan decomposition 
\begin{equation}
	\widehat{\mathfrak{sl}}_2=\mathfrak{n}_-\oplus\mathfrak{h}\oplus\mathfrak{n}_+,
		\label{eq:cartdec}
\end{equation}
while the definition
\[\mathfrak{k}=\C J^-_0\oplus\C J^+_0\oplus\mathfrak{h} = \mathfrak{sl}_2\oplus\C K\oplus\C d,\] yields another triangular decomposition:
\begin{equation}
	\widehat{\mathfrak{sl}}_2=\mathfrak{r}_-\oplus\mathfrak{k}\oplus\mathfrak{r}_+.
	\label{eq:triangdec}
\end{equation}
\begin{defn}\label{def:prolmod}
Let $V$ be a $\mathcal{U}(\mathfrak{sl}_2)$--module. Fix  a \emph{level} $k\in\C\backslash\{-2\}$ and extend $V$ to a $\mathcal{U}(\mathfrak{k})$ module $V^{(k)}$ by letting $K$ and $d$ act as $k\, \mathrm{id}_V$ and $\frac{C^{(2)}}{k+2}$ respectively, where $C^{(2)}$ denotes the second order Casimir of $\mathfrak{sl}_2$. Demanding that $\mathcal{U}(\mathfrak{r}_+)$ annihilates $V^{(k)}$, it becomes a $\mathcal{U}(\mathfrak{k}\oplus\mathfrak{r}_+)$--module. Now, define the $\mathcal{U}(\widehat{\mathfrak{sl}}_2)$--module $V^k$ as
\begin{equation}
	V^k \defas \mathcal{U}(\widehat{\mathfrak{sl}}_2)\otimes_{\mathcal{U}(\mathfrak{k}\oplus\mathfrak{r}_+)}V^{(k)}.
	\label{eq:prolmod}
\end{equation}
\end{defn}
We refer to modules of the type $V^k$ as prolongation modules (being prolonged from $V$).
\begin{rem}~
\begin{itemize}
	\item[(i)] As usual the Sugawara construction defines an action of the (enveloping algebra of the) Virasoro algebra on $V^k$, such that $L_0$ acts as $-d$.
	\item[(ii)] Note that the definition above implies $V^k=\mathcal{U}(\mathfrak{r}_-)\oti_\C V^{(k)}$. In particular, $V^k$ is free as a $\mathcal{U}(\mathfrak{r}_-)$--module.
\end{itemize}
\end{rem}

In the literature on the $SL(2,\R)$ WZW model an infinite cyclic subgroup of $\mathrm{Aut}(\widehat{\mathfrak{sl}}_2)$, the group of so called spectral flow automorphisms, plays an important role. 
\begin{defn}\label{def:sflow}
	For $s\in\Z$ the spectral flow automorphism $\vartheta_s$ is defined through
\begin{align}
	J^\pm_{n} &\mapsto J^\pm_{m\mp s}\\
	J^0_n &\mapsto J^0_n - s K\delta_n\\
	K & \mapsto K\\
	d &\mapsto d + \frac{s}{2}J^0_0 - \frac{s^2}{4} K\label{eq:sfdef}
\end{align}
\end{defn}
The spectral flow automorphisms extend automorphically to $\mathcal{U}(\widehat{\mathfrak{sl}}_2)$, denote these extensions by the same symbols $\vartheta_s$.
\begin{rem}~
\begin{itemize}
	\item[(i)] By the Sugawara construction spectral flow also acts automorphically on the Virasoro algebra as
			\[\vartheta_s: L_n \mapsto L_n - \frac{s}{2}J^0_n + \frac{s^2}{4}K\delta_n.\]
	\item[(ii)] It follows easily that $\vartheta_s\circ\vartheta_t = \vartheta_{s+t}$, so the infinite cyclic group of spectral flow automorphisms is $\langle\vartheta_1\rangle\subset\mathrm{Aut}(\widehat{\mathfrak{sl}}_2)$. 
	\item[(iii)] Note that $\vartheta_s\circ\omega=\omega\circ\vartheta_s$. 
\end{itemize}
\end{rem}
As automorphisms of $\mathcal{U}(\widehat{\mathfrak{sl}}_2)$, the $\vartheta_s$ also act on the corresponding category of modules.
\begin{defn}\label{def:sfmod}
For any $\mathcal{U}(\widehat{\mathfrak{sl}}_2)$--module $V^k\equiv(V^k,\rho)$, where $\rho:\mathcal{U}(\widehat{\mathfrak{sl}}_2)\rightarrow\mathrm{End}(V)$ denotes the representation morphism, and any $s\in \Z$, define the module $V^{k,s}$ as
	\begin{equation}
		V^{k,s}\equiv(V^k,\rho^s)\defas (V^k,\rho\circ\vartheta_s).
		\label{eq:sfmod}
	\end{equation}
\end{defn}
Restrict the level $k$ for the modules under consideration such that $k+2<0$ and $k\in\Z$, and fix a $k$ satisfying these conditions. Denote the corresponding chiral category by $\mathcal{C}^k$. The motivation for considering negative levels comes from the no ghost theorem \cite{Hwang}, from which it follows that there is then a set of modules constructed from the discrete and principal continuous series of $sl(2,\R)$ leading to a unitary \emph{string} spectrum. Although the object under study here is a CFT and not a string theory, it is reasonable to first study the case that is at present best understood.
The motivation behind the restriction to integer $k$ is twofold. First, with this restriction there is an existing proposal for a modular invariant bulk state space \cite{HeHwRS}, where the spectrum appears to be related to the actual group $SL(2,\R)$ and not to any cover thereof. Second, an analysis of Chern--Simons gauge theory with gauge group $SL(2,\R)$ (or indeed any finite cover thereof) reveals a quantization condition on the (Chern--Simons) level \cite{Wit}. A relation between the $SL(2,\R)$ WZW model and the corresponding Chern--Simons theory analogous to the compact case then directly leads to a quantization condition on the level.\\

The category $\mathcal{C}^k$ is expected to contain objects corresponding to $\mathcal{U}(\widehat{\mathfrak{sl}}_2)$--modules constructed from the discrete series $\{D^\pm_j\}_j$ of highest and lowest weight $\mathcal{U}(\mathfrak{sl}_2)$--modules, and from the principal continuous series $\{C^\alpha_j\}_{\alpha,j}$.
This follows from a path integral formulation where the zeromodes take values in $L^2(SL(2,\R))$, which decomposes in unitary representations of discrete and principal continuous type (see for instance \cite{Lang}). For the modules coming from the discrete series, assume those satisfying the bound $\frac{k+2}{2}\leq j\leq -1$ appear in $\mathcal{C}^k$ (here $j$ denotes half the $J^0$ eigenvalue). In addition modules originating in the limits of the discrete series with $j=-\frac{1}{2}$ may appear (these, however, do not appear in the decomposition of $L^2(SL(2,\R))$), as well as possibly modules corresponding to $j=\frac{k+1}{2}$. This bound is again motivated by the no ghost theorem of \cite{Hwang}, which gives the bound $j>k/2$. Also, for $k\in \Z$ modules prolonged from $D^\pm_j$ are void of null vectors iff $j>k/2$ \cite{Peskin}, and an argument analogous to that in \cite{Gepner} for compact groups results in the decoupling in amplitudes of fields corresponding to modules \emph{with} null vectors.
In addition to the class of modules discussed above one expects $\mathcal{C}^k$ to contain all modules obtained from these by the action of spectral flow \cite{HeHwRS,MaOg}.

Although this set of modules cover the proposed modular invariant spectra, there are reasons to believe that $\mathcal{C}^k$ should contain also other objects. Under the assumption that it is a tensor category there should also be a tensor unit, and such an object is not to be found among the ones above. From the point of view of CFT, even though there is no identity field in the spectrum (the trivial representation of $sl(2,\R)$ is not found in $L^2(SL(2,\R))$ due to the noncompactness), this field should nevertheless play a role in an operator formalism. Denote the one dimensional trivial representation of $\mathfrak{sl}_2$ by $D_0$.
A comparison with the compact case then leads to the conjecture that $\mathcal{C}^k$ should contain the vacuum module $D^k_0$. We also conjecture that $\mathcal{C}^k$ contains modules obtained from the vacuum module by the action of spectral flow, i.e. $D^{k,s}_0$, $s\in\Z$.\\

The origins as a representation category leads naturally to $\mathcal{C}^k$ being $\C$--linear and abelian.
As already declared, one expects $\mathcal{C}^k$ to be a tensor category. It is indeed hard to see how any of the familiar structure of rational CFT will generalize unless there is a fusion product, associative in a suitable sense, of objects in $\mathcal{C}^k$. 
In the rational case the ribbon structure of the chiral category is directly related to properties of the bundles of conformal blocks and the projectively flat connections on these. The construction of vector spaces of conformal blocks also makes sense in the $SL(2,\R)$ case. In addition, the physical origin of the coinvariance condition, the chiral Ward identities, will still have to hold in a conventional quantum field theory. The definition of spaces of conformal blocks, at least on $\C\mathbb{P}^1$, will therefore be directly adopted from Definition \ref{def:confblock}. 
In addition we will make the crucial assumption that Proposition \ref{prop:confblockprops} $(v)$ still holds in the $SL(2,\R)$ case, i.e. the following relation between spaces of conformal blocks on $\C\mathbb{P}^1$ and hom spaces in $\mathcal{C}^k$ is assumed to hold:
\begin{equation}
	\mathcal{H}(\vec{p},V^k_{\vec{\lambda}})\cong\Hom_{\mathcal{C}^k}(V^k_{\lambda_1}\fuse \cdots\fuse V^k_{\lambda_n},\one)=\Hom_{\mathcal{C}^k}(V^k_{\lambda_1}\fuse \cdots\fuse V^k_{\lambda_n},D^k_0),
	\label{eq:confblockHom}
\end{equation}
for any $n$ tuple $(V_{\lambda_1},\ldots,V_{\lambda_n})$ of $\mathcal{U}(\widehat{\mathfrak{sl}}_2)$--modules in $\mathcal{C}^k$.
Independence of the choice of local coordinates in the definition of $\mathcal{H}(\vec{p},V^k_{\vec{\lambda}})$ only depends on the Sugawara construction, and will still hold. Furthermore, chiral correlators on $\C\mathbb{P}^1$ will still satisfy Knizhnik--Zamolodchikov equations, so it is natural to expect that spaces of conformal blocks will also in the $SL(2,\R)$ case combine to vector bundles of conformal blocks, equipped with flat connections, over the moduli space of spheres with marked points.
We will simply assume that this indeed holds, and that consequently the vector spaces $\mathcal{H}(\vec{p},V^k_{\vec{\lambda}})$ are, up to isomorphism, independent of the position of the insertion points. The monodromies of the corresponding connections are then expected to equip $\mathcal{C}^k$ with a braiding, and a structure analogous to the balancing (i.e. the family of twist automorphisms $\Theta_X$) in the rational case.

It may be possible to define a conjugation on primary fields through the spaces of two--point conformal blocks on $\C\mathbb{P}^1$, in the case that these define a nondegenerate pairing of primary fields. In fact, for primary fields corresponding to the discrete series, $\{D^\pm_j\}$, the two--point functions appearing in the literature, \cite{MaOg2} via analytical continuation of the $H^+_3$ model results \cite{Teschner1,Teschner2}, and \cite{Nun} via a free field construction, indicate that the conjugate of $D^{+,k,s}_j$ is $D^{-,k,-s}_j$. Note, however, that the classic result \cite{Puk,Repka}
\begin{equation}
	D^+_j\oti D^-_j\cong \int^{\oplus}_{\R}d\lambda\, C^0_{-\frac{1}{2}+i\lambda},
	\label{eq:sl2rtensor}
\end{equation}
as well as the proposed fusion rules \cite{Nun2} of the $SL(2,\R)$ model (which essentially amounts to a truncation of \eqref{eq:sl2rtensor} with modifications due to spectral flow) appears to be in contradiction with the assumption \eqref{eq:confblockHom}, and indeed with any sensible operator formulation, due to the nonappearance of the trivial module in \eqref{eq:sl2rtensor}. The apparent contradiction relies on a conventional interpretation of the fusion rules which assumes that $\mathcal{C}^k$ is semisimple. As will be shown further down, properties appearing already for $\mathcal{U}(\mathfrak{sl}_2)$--modules provide strong evidence that $\mathcal{C}^k$ is not semisimple, and that the fusion product will in general not be fully decomposable.

The category $\mathcal{C}^k$ can thus not be expected to be semisimple. As a consequence (compare the reconstruction of a weakly ribbon category from Moore--Seiberg data in \cite{BK}) there is no strong reason to believe that the conjugation on primary fields will equip $\mathcal{C}^k$ with a duality. It may very well happen, however, that duality is replaced with a weaker structure, see Remark \ref{rem:semirig}.\\

Among nonchiral aspects, one issue will be discussed here, namely the bulk state space and torus partition function.
Recall from \cite{MaOg} that the bulk state space of the theory on the infinite cover was proposed to be obtained from sectors of the form\footnote{The precise expression in \cite{MaOg} actually differs by replacing in the second factor $-\mapsto +$ and $-s\mapsto s$. This discrepancy is purely a matter of convention, however, and related to the two conventions appearing in the literature concerning the meaning of ``charge conjugation'' vs. ``diagonal'' modular invariants.} $D^{+,k,s}_{j}\otimes_\C D^{-,k,-s}_{j}$ and $C^{\alpha,k,s}_{j}\otimes_\C C^{\alpha,k,-s}_{j}$ by taking the direct sum over all integer values of $s$ and the direct integral over all values of $j$ satisfying the unitarity bound. This bulk state space leads to a torus partition function that begs to be called the charge conjugation torus partition function. A comparison with two point functions on the sphere \cite{MaOg2,Nun} further strengthens this impression.

On the single cover the results of \cite{HeHwRS} indicate that one might construct bulk state spaces from bricks of chiral objects of the form
\begin{equation}
	X^k_j\defas X_j^{+,k}\oplus X_j^{-,k},
	\label{eq:brick1}
\end{equation}
where
\begin{equation}
	X_j^{\pm,k} \defas \bigoplus_{s\in 2\Z}D^{\pm,k,s}_j,
	\label{eq:brick2}
\end{equation}
and where $j$ takes half integer values and satisfies the unitarity bound (only the discrete series were discussed). The multiplicities of the various combinations $X^k_i\otimes_\C X^k_j$ would, according to the proposal of \cite{HeHwRS}, follow the $SU(2)$ series of modular invariants at level $-k-4$\footnote{The condition on the level was not explicitly stated, but can be straightforwardly derived from the results. Note that to get a semi positive level in the $SU(2)$ case one must require $k+4\leq 0$, which is stronger than previously required. The reason is that the result of summing up the characters of the modules contained in $X^k_j$ for $j=-1/2$ and $j=(k+1)/2$ is, after discarding a delta function, zero. To get a nonvanishing contribution (again, after discarding certain delta functions) the brick $X^k_{-1}$ must be included and $-1>(k+1)/2$, requiring $k+2\leq -2$. See also a more thorough discussion in \cite{BF}}. This proposal can with similar arguments be extended to $k\in\mathbb{Q}$, with the interpretation of the sigma model on a finite cover of $SL(2,\R)$ \cite{FjHw}, leading to an intermediate between (products of) the single cover bricks and the infinite cover sectors.
Note that since each brick $X^k_j$ contains both $D^{+,k,s}_j$ and its reputed conjugate $D^{-,k,-s}_j$ it is tempting to say that $X^k_j$ is self conjugate, which further strengthens the analogy with the $SU(2)$ WZW model where simple objects are self dual. Recall, however, that $\mathcal{C}^k$ will likely not have a duality, so the meaning of a primary field being self conjugate (beyond having nonvanishing two--point conformal block with itself on $\C\mathbb{P}^1$) is unclear.
Both on a finite cover and on the infinite cover one may thus identify, in particular, theories with a charge conjugation bulk spectrum. Importantly, however, none of these cases correspond to a Cardy case theory since the bulk spectrum does not contain a vacuum sector.\\

The appearance of bricks like $X^k_j$ in the bulk spectrum is in the rational case a signal of an extended chiral algebra, which can in turn be identified from the vacuum brick. In analogy with rational CFT it is natural to ask whether $\mathcal{C}^k$ contains invertible objects such that the bricks $X^k_j$ are the result of a simple current extension. If that is the case then obviously there must exist simple currents with infinite order. Again comparing to the rational case this is natural. In every compact (rational) WZW model, except for the $E_8$ model at level $2$, the Picard group of the chiral category is isomorphic to the center of the simply connected real form of the corresponding Lie group. The center of $SL(2,\R)$ is an infinite cyclic group, so it is not inconceivable to find an invertible object generating an infinite cyclic subgroup of $\mathrm{Pic}(\mathcal{C}^k)$. In further analogy with the rational $SU(2)$ models one might expect that the action of spectral flow on the vacuum module generates invertible objects.
It should be pointed out that there have already appeared, in a free field realization \cite{Nun3}, ``spectral flow operators'' whose insertions have an effect similar to insertions of the purported simple currents associated to spectral flow. Comparing the charges of these operators, however, it appears they do not correspond exactly to modules obtained by acting with spectral flow on the vacuum module.

%%%%%%%%%%%%%%%%%%%%%%%%%%%%%%%%%%%%%%%%%%%%%%%%%%%%%%%%%
\section{Modules related to the discrete series and their contragredients}\label{sec:mods}
Let us first fix some notation. The algebraic tensor product of vector spaces will be denoted by $\oti$, whereas the Hilbert space tensor product of Hilbert spaces will be denoted by $\woti$. In other words, $\woti$ involves a completion w.r.t. the Hilbert space norm compared to $\oti$. The same notation is used to distinguish the Hilbert space direct sum, $\wop$, from the vector space direct sum, $\oplus$. A similar notation is used to distinguish \emph{unitary} modules ${}^\mathrm{H}V=\left(V,(\cdot,\cdot),\rho\right)$ where $\left(V,(\cdot,\cdot)\right)$ is a $\C$--Hilbert space and $\rho:\env{g}\rightarrow \mathcal{U}\left(\mathfrak{u}\left(V,(\cdot,\cdot)\right)\right)$ is the representation morphism, from \emph{linear} modules $V=(V,\rho)$ where $V$ is a $\C$--vector space and $\rho:\env{g}\rightarrow\End(V)$ is the representation morphism.

%%%%%%%%%%%%%%%%%%%%%%%%%%%%%%%%%%%%%%%%%%%%%
\subsection{The discrete series of \texorpdfstring{$sl(2,\R)$}{sl(2,R)}}\label{sec:sl2discr}
The (very basic) results presented in this subsection are not claimed to be original. However, to the best of our knowledge the main result, Proposition \ref{prop:sl2nonsemi}, does not appear explicitly in the literature. The distinction between unitary and unitarizable modules is emphasized, and the corresponding results give valuable hints concerning the structure of the category $\mathcal{C}^k$.
Let, as in the introduction, $sl(2,\R)$ mean the pair $(\mathfrak{sl}_2,\widetilde{\omega})$, where $\mathfrak{sl}_2$ satisfies \eqref{sl2}. Recall that $L^2(SL(2,\R))$ can be decomposed into the discrete series and the principal continuous series of unitary irreducible representations of $SL(2,\R)$ (and of $sl(2,\R)$). The discrete series, $ \uD^\pm_j$ where $j\in -\frac{1}{2}\N$, $j\leq -1$, contains lowest weight representations ($ \uD^+_j$) and highest weight representations ($ \uD^-_j$). In addition there are the unitary irreducible mock discrete (or limits of discrete series) representations $ \uD^\pm_{-\frac{1}{2}}$. We will treat the latter representations as part of the discrete series.

Fix $j\in -\frac{1}{2}\N$. The $\mathcal{U}(\mathfrak{sl}_2)$--module $ \uD^+_j$ is a Hilbert space with an orthonormal basis $\{|j,m\rangle | m\in\N_0\}$ such that
\begin{align}
	J^0|j,m\rangle &= 2(-j+m)|j,m\rangle\\
	J^+|j,m\rangle &= \sqrt{(m+1)(m-2j)}|j,m+1\rangle\\
	J^-|j,m\rangle &= -\sqrt{m(m-1-2j)}|j,m-1\rangle.
\end{align}
Denote the inner product by $(\cdot,\cdot)$. The module $ \uD^+_j$ can be expressed in terms of a weight decomposition 
\[  \uD^+_j=\bigoplus^\mathrm{H}_{m\in\N_0}D^+_{j,m} :=\bigoplus^\mathrm{H}_{m\in\N_0}\C |j,m\rangle. \]
The algebra $\mathfrak{sl}_2$ is represented unitarily on this Hilbert space with the adjoint given by $\widetilde{\omega}$.
Denote by $D^+_j$ the underlying, invariant and simple, weight module
\[D^+_j=\bigoplus_{m\in\N_0}D^+_{j,m}.\]
Equip, by restriction, $D^+_j$ with the inner product $(\cdot,\cdot)$. Of course $D^+_j$ is unitarizable, and the unitarization is equivalent to the completion w.r.t. the inner product norm which by definition is $ \uD^+_j$.

The module $ \uD^-_j$ is analogously a Hilbert space with an orthonormal basis $\{|j,m\rangle|m\in\N_0\}$ (if needed, bases of $D_j^\pm$ will be distinguished by a subscript, $\{|j,m\rangle_\pm\}$) such that
\begin{align}
	J^0|j,m\rangle &= 2(j-m)|j,m\rangle\\
	J^+|j,m\rangle &= -\sqrt{m(m-1-2j)}|j,m-1\rangle\\
	J^-|j,m\rangle &= \sqrt{(m+1)(m-2j)}|j,m+1\rangle.
\end{align}
The structure described for $ \uD^+_j$ and $D^+_j$ extends immediately to
\[ \uD^-_j=\bigoplus^\mathrm{H}_{m\in\N_0}D^-_{j,m}\]
and to (the simple) $D^-_j$.\\

\noindent
Denote by $( \uD^\pm_j)^\dagger$ the module conjugate to $ \uD^\pm_j$, it is well known and easily shown that $( \uD^\pm_j)^\dagger\cong  \uD^\mp_j$.
Corresponding to the weight modules $D^\pm_j$, one can define their dual modules $(D^\pm_j)^*$. Since $D^\pm_j$ are infinite dimensional, however, their duals are reducible. Instead, make use of the weight decomposition to define contragredient modules.
\begin{defn}
	The module $(D^\pm_j)^\cont$ contragredient to $D^\pm_j$ is defined on the restricted dual space
	\[(D^\pm_j)^\cont :=\bigoplus_{m\in\N_0}(D_{j,m}^\pm)^*\subsetneq (D^\pm_j)^*.\]
	The representation morphism on the contragredient module is defined by
	\[ (J^a\varphi)(v) = -\varphi(J^a v).\]
\end{defn}
Modules contragredient to $D^\pm_j$ satisfy the following obvious property.
\begin{prop}
	\[\big(D^\pm_j\big)^\cont\cong D^\mp_j\]
\end{prop}
\begin{proof}
	The proof is trivial. Defining the dual bases $\{\varphi^\pm_m\}_{m\in\N_0}$ of the contragredient modules by $\varphi^\pm_m(|j,n\rangle)=\delta_{m,n}$ we have on $(D^+_j)^\cont$:
\begin{align}
	J^0\varphi^+_m &= 2(j-m)\varphi^+_m\\
	J^+\varphi^+_m &= -\sqrt{m(m-1-2j)}\varphi^+_{m-1}\\
	J^-\varphi^+_m &= \sqrt{(m+1)(m-2j)}\varphi^+_{m+1}.
\end{align}
The linear map $\varphi^+_m\mapsto |j,m\rangle_-$ obviously defines an intertwining isomorphism from $( D^+_j)^\cont$ to $ D^-_j$. There is another equally obvious intertwining isomorphism from $( D^-_j)^\cont$ to $ D^+_j$.
\end{proof}
The proof of the following proposition is equally trivial and will not be spelled out.
\begin{prop}~
\begin{itemize}
	\item[(i)] The map
		\begin{equation}
			\beta_j: D^+_j\times D^-_j\rightarrow\C,\ (|j,m\rangle_+,|j,n\rangle_-)\mapsto \delta_{m,n}
		\end{equation}
		defines a, unique up to a scale factor,  (nondegenerate) invariant bilinear pairing.
	\item[(ii)]  Consequently the linear map $\psi_j: D^+_j\otimes D^-_j\rightarrow\C$ defined by
		\begin{equation}
			\psi_j:\sum_{m,n} c_{m,n}|j,m\rangle_+\otimes |j,n\rangle_-\mapsto \sum_m c_{m,m}\label{psi}
		\end{equation}
		gives the, unique up to a scale factor, nonzero intertwiner from $ D^+_j\otimes D^-_j$ to the trivial module $D_0\cong\C$.
\end{itemize}
\end{prop}
It is not difficult to show that the intertwiner $\psi_j$ has no section. In fact, the only possible nonzero intertwiner from the trivial representation to $ D^+_j\otimes D^-_j$ would have image $\C\sum_{m=0}^\infty |j,m\rangle_+\otimes |j,m\rangle_-$. Furthermore, one quickly verifies that there are no intertwiners in any direction between $D^+_{j_1}\oti D^-_{j_2}$ and $D_0$ whenever $j_1\neq j_2$. These results are summarized in the following
\begin{prop}\label{prop:sl2nonsemi}
	\begin{align}
		\Hom_{\ens}\left(D^\pm_i\oti\big(D^\pm_j\big)^\cont,D_0\right) \cong \Hom_{\ens}\left(D^\pm_i\oti D^\mp_j,D_0\right) &\cong \delta_{i,j}\C\label{eq:sl2nonss1}\\
		\Hom_{\ens}\left(D_0,D^\pm_i\oti\big(D^\pm_j\big)^\cont\right) \cong \Hom_{\ens}\left(D_0,D^\pm_i\oti D^\mp_j\right) & = \{0\}.\label{eq:sl2nonss2}
	\end{align}
\end{prop}

It is immediately clear that $\psi_j$ does not extend to $ \uD^+_j\woti  \uD^-_j$ since it is not defined on e.g. $\sum_{m=0}^\infty \frac{1}{m}|j,m\rangle_+\otimes |j,m\rangle_-\in  \uD^+_j\woti  \uD^-_j$. Indeed, the decomposition \eqref{eq:sl2rtensor}
implies that there are no (continuous) intertwiners from $ \uD^+_j\woti  \uD^-_j$ to the trivial representation, nor in the other direction. Denoting by $\widehat{\Hom}$ the continuous homomorphisms of unitary representations, the results above are summarized as
\begin{align}
	\widehat{\Hom}_{sl(2,\R)}\left(\uD^\pm_j\woti(\uD^\pm_j)^\dagger,\uD_0\right) \cong \widehat{\Hom}_{sl(2,\R)}\left(\uD^\pm_j\woti \uD^\mp_j,\uD_0\right) &= \{0\}\label{eq:unhom1}\\
	\widehat{\Hom}_{sl(2,\R)}\left(\uD_0,\uD^\pm_j\woti(\uD^\pm_j)^\dagger\right) \cong \widehat{\Hom}_{sl(2,\R)}\left(\uD_0,\uD^\pm_j\woti\uD^\mp_j\right) &= \{0\}\label{eq:unhom2}
\end{align}
In these formulas, $\uD_0$ is viewed as a unitary module with the obvious inner product.\\
The trivial, but nevertheless interesting, observations in Proposition \ref{prop:sl2nonsemi} and equations \eqref{eq:unhom1} and \eqref{eq:unhom2} illustrate that categories of unitarizable and unitary representations may have very different properties in the infinite dimensional case.

\begin{rem}\label{rem:semirig}
	An immediate implication of Proposition \ref{prop:sl2nonsemi} is that a monoidal subcategory of $\mathrm{Rep}(\ens)$ containing both $D^\pm_j$ for at least one $j$ can be neither semisimple nor rigid. It is conceivable that such categories may still possess a particular generalization of rigidity so called \emph{semi--rigidity} \cite{Miyamoto}. Consistent with equations \eqref{eq:sl2nonss1} and \eqref{eq:sl2nonss2}, the definition of semi rigidity includes morphisms \emph{from} objects of the type $D\oti \widetilde{D}$ \emph{to} $\one$, but not in the opposite direction.
\end{rem}

The reason why the results in this section are interesting for us is the connection between the fusion product and the tensor product of modules of the horizontal subalgebra. In particular, if primary fields correspond to modules prolonged from $D^\pm_j$ one is lead to expect that the structure of hom spaces in Proposition \ref{prop:sl2nonsemi} will lift to corresponding hom spaces in the chiral category $\mathcal{C}^k$.

%%%%%%%%%%%%%%%%%%%%%%%%%%%%%%%%%%%%%%%%%
\subsection{\texorpdfstring{$\mathcal{U}(\widehat{\mathfrak{sl}}_2)$}{U(affine sl2)} avatars of the discrete series and the trivial module}\label{sec:proldiscr}
Using the constructions in Definitions \ref{def:prolmod} and \ref{def:sfmod}, the $\mathcal{U}(\widehat{\mathfrak{sl}}_2)$--modules $\uD^{\pm,k,s}_{j}$, $D^{\pm,k,s}_j$, and $D^{k,s}_0$ may be constructed. Let us begin by conjecturing that
\begin{quote}
	\emph{Among the objects of $\mathcal{C}^k$ may be found $\mathcal{U}(\widehat{\mathfrak{sl}}_2)$--modules of the type $D^{\pm,k,s}_j$, but not of the type $\uD^{\pm,k,s}_j$.}
\end{quote}
The main reason to believe this conjecture is the difficulty to reconcile known two--point functions with \eqref{eq:unhom1} and \eqref{eq:unhom2}. In the case of infinite cover we know \cite{Nun} that conjugate primary fields coincide (for vanishing spectral flow) with conjugate representations of $sl(2,\R)$. If the OPE of the two primary fields reflects the fusion rules, equations \eqref{eq:unhom1} and \eqref{eq:unhom2} implies that the identity field does not appear in the OPE, and nonvanishing of the two--point function appears mysterious\footnote{One might suspect that there is instead another field with similar properties to the identity, but the unitary Clebsch--Gordan series of $sl(2,\R)$ implies that there should not appear any primary field in the OPE with vanishing conformal dimension.}. The structure in Proposition \ref{prop:sl2nonsemi}, on the other hand, rhymes well with the nonvanishing two--point functions.

In addition it seems like a very unnatural procedure to prolong a unitary $\mathcal{U}(\mathfrak{sl}_2)$--module only linearly. To cure this one could try to unitarize the complete module $\uD^{\pm,k}_j$, but the unitary structure on $\uD^\pm_j$ does not extend to a unitary structure on $\uD^{\pm,k}_j$ w.r.t. $\omega$ (the induced sesquilinear form is not positive definite). As will be shown further down it is in even impossible to find an inner product on ${D}^{\pm,k}_j\subset \uD^{\pm,k}_j$ that unitarizes the horizontal subalgebra w.r.t. $\widetilde{\omega}$.\\

It is clear from the definition that $ D^\pm_j$ are Verma modules of $\mathcal{U}(\mathfrak{sl}_2)$ (extending the notion of Verma module to cover also lowest weight modules). Using the Cartan decomposition of $\widehat{\mathfrak{sl}}_2$ it is therefore also obvious that we can write
\[{D}^{-,k}_j\defas \mathcal{U}(\widehat{\mathfrak{sl}}_2)\oti_{\mathcal{U}(\mathfrak{k}+\mathfrak{r}_+)} D^{-,k}_j=\mathcal{U}(\widehat{\mathfrak{sl}}_2)\oti_{\mathcal{U}(\mathfrak{h}\oplus\mathfrak{n}_+)} \C |j,0\rangle,\]
and, with a suitable modification of $\mathfrak{n}_+$, an analogous expression for ${D}^{+,k}_j$. It follows that ${D}^{\pm,k}_j$ are Verma modules, as is of course well known. Note that this is \emph{not} the case for $ \uD^{\pm,k}_j$.\\

\noindent
The restriction to ${D}^{\pm,k,0}_j$ of the following result first appeared in \cite{Peskin}.
\begin{prop}
	For $k<0$, the $\mathcal{U}(\widehat{\mathfrak{sl}}_2)$--modules ${D}^{k,s}_0$ are simple, and the modules ${D}^{\pm,k,s}_j$ are simple iff $j>k/2$.
\end{prop}
\begin{proof}
The proof follows from the form of the Kac--Kazhdan determinant of the contravariant forms of the respective Verma modules ${D}^{k,0}_0$ and ${D}^{\pm,k,0}_j$.
Recall that the Kac--Kazhdan determinant \cite{KacKaz} for the subspace of weight $\lambda-\eta$ of a $\mathcal{U}(\widehat{\mathfrak{sl}}_2)$ Verma module of highest weight $\lambda$ is given by
\begin{equation}
	\mathrm{det}(K_\eta) = \prod_{\alpha\in\Delta^+}\prod_{n\in\N}\left[ (\lambda+\rho,\alpha)- \frac{n}{2}(\alpha,\alpha) \right]^{P(\eta-n\alpha)}, \label{KKdet}
\end{equation}
where $\Delta^+$ is the set of positive roots of $\widehat{\mathfrak{sl}}_2$, $\rho$ is the affine Weyl vector, and $P$ is the Kostant partition function.
Consider first ${D}^{-,k}_j$. This is a Verma module with highest weight $\lambda=(k-2j)\Lambda_0+2j\Lambda_1$, where $\Lambda_0$, $\Lambda_1$ are the fundamental weights (the roots of $\mathfrak{sl}_2$ are normalized to have length squared $2$). The most stringent requirement comes from the factor with $n=1$ and the simple root $\alpha_0=\alpha_0^\vee$:
\[ (\lambda+\rho,\alpha_0)-\frac{1}{2}(\alpha_0,\alpha_0)=(k-2j)+1-1=k-2j.\]
It follows that the kernel of the contravariant form is trivial, and there are therefore no proper submodules, iff $j>k/2$ (if $j\leq k/2$ there is always a root that gives a factor of $0$). The result is the same for ${D}^{+,k}_j$, with a suitable modification of the triangular decomposition to get a highest weight module.
Obviously ${D}^{\pm,k,s}_j$ is simple iff ${D}^{\pm,k,0}_j={D}^{\pm,k}_j$ is simple, so the same statement holds for the modules obtained from ${D}^{\pm,k}_j$ by spectral flow.  Also ${D}_0^k$ may be analyzed in the same way since it is a Verma module with highest weight $k\Lambda_0$. The (well known) result is that ${D}_0^k$ is simple iff $k\notin\N_0$, in particular it is simple when $k<0$. Obviously then also the modules ${D}^{k,s}_0$ are simple.
\end{proof}
\begin{rem}
	For generic $s$,  ${D}^{\pm,k,s}_j$ and ${D}^{k,s}_0$ are no longer highest/lowest weight modules. There is, however, for every $s$ a cyclic vector from which the whole module is generated, namely the image of the highest weight vector under the map ${D}^{\pm,k}_j\mapsto {D}^{\pm,k,s}_j$.
	In ${D}^{\pm,k,\mp 1}_j$ the cyclic vectors are actually lowest/highest weight vectors, and that the same is true for ${D}^{k,1}_0$ respectively ${D}^{k,-1}_0$.
\end{rem}

The following result is well known, the first statement appearing in \cite{MaOg}.
\begin{prop}\label{prop:sfequiv}
	For $k<0$, $j>k/2$, there are the following equivalences ${D}^{+,k,s-1}_j\cong{D}^{-,k,s}_{k/2-j}$. There are no other equivalences among the modules $D^{\pm,k,s}_j$, $D^{k,s}_0$.
\end{prop}
\begin{proof}
Since all the modules under consideration are simple and generated by a cyclic vector the proof is straightforward by checking the action of spectral flow on the cyclic vector.
The weights (more precisely, the ($L_0,J^0_0$) eigenvalues) occurring in these modules are:
\begin{align}
	{D}^{+,k,s}_{j} &: \left\{ \left(\frac{j(j+1)}{k+2}\!+\!\frac{s^2k}{4}\! -\! (-j+n)s\!+\! m,2n\!-\! 2j\!-\! sk\right)\bigg|m\in\N_0, n\in\Z_{\geq -m}\right\}\label{wts1}\\
	{D}^{-,k,s}_{j}&: \left\{ \left(\frac{j(j+1)}{k+2}\!+\!\frac{s^2k}{4}\!+\!(-j+n)s\!+\!m,2j\!-\!sk\!-\!2n\right)\bigg|m\in \N_0, n\in \Z_{\geq -m}\right\}\label{wts2}\\
	{D}^{k,s}_0 &:\left\{ \left(m+\frac{s^2k}{4}-sn,2n-sk\right)\bigg| m\in\N_0, n=-m,-m+1,\ldots,m\right\}.
\end{align}
The set $\{\frac{k+1}{2},\frac{k+2}{2},\ldots,-\frac{1}{2}\}$ of allowed values of $j$ is preserved by the reflection $j\mapsto k/2-j$. Note that, since \[\frac{(k/2-j)(k/2-j+1)}{k+2}=\frac{j(j+1)}{k+2}+\frac{k}{4}-j,\] the weights occurring in ${D}^{-,k,0}_{k/2-j}$ coincide with those occurring in ${D}^{+,k,-1}_{j}$. Since ${D}^{+,k,-1}_{j}$ is a simple highest weight module with the same highest weight as ${D}^{-,k,0}_{k/2-j}$, these two modules must be equivalent. This equivalence immediately generalizes to the equivalences ${D}^{+,k,s-1}_{j}\cong{D}^{-,k,s}_{k/2-j}$, and inspection of \eqref{wts1} and \eqref{wts2} implies that there are no other equivalences within this class of modules. It is straightforward to check that there are no equivalences within the class ${D}^{k,s}_0$.
\end{proof}

\noindent
If $\mathcal{U}(\widehat{\mathfrak{sl}}_2)$--modules are given the conventional $\Z$ grading\footnote{Actually, one is usually considering an $\N_0$ grading. To treat all the modules of interest here on the same footing it will be convenient to use a $\Z$ grading where, if necessary, a module is completed with trivial subspaces $\{0\}$ in all grades not covered by the $\N_0$ grading.} using $L_0$, the nonvanishing homogeneous subspaces of the modules discussed here, with the exception of $D^k_0$, are all infinite dimensional. Furthermore, with the exceptions ${D}^{k,0}_0$, ${D}^{k,\pm1}_0$, ${D}^{\pm,k,0}_j$, ${D}^{+,k,-1}_j$, ${D}^{-,k,1}_j$, homogeneous subspaces are nontrivial at \emph{all} $\Z$--grades.
	
The weight decomposition can in fact be used to provide a $\Z\times\Z$ grading, where the first factor corresponds to the integer part of the $L_0$ eigenvalue and the second factor corresponds to the $J^0_0$ eigenvalue. Fix the following convention: the homogeneous subspace labeled by $(0,0)\in\Z\times\Z$ is spanned by the distinguished cyclic vector (i.e. the image of the highest/lowest weight vector). This determines the grading in terms of $(L_0,J^0_0)$ eigenvalues uniquely.
The importance of the $\Z\times\Z$ grading comes from the following, trivially shown, properties.
\begin{prop}\label{prop:grading}
For any $V\in\{D^{k,s}_0,D^{\pm,k,s}_j\}_{s,j}$ indicate the $\Z\times\Z$ grading as
\[V=\bigoplus_{m,n\in\Z}V_{m,n}.\]
Then
\begin{itemize}
	\item[(i)] $\mathrm{dim}_\C(V_{m,n})<\infty$ for all $m,n\in\Z$.
	\item[(ii)] For any $n\in\Z$ there exists a $N_n\in \Z$ such that $V_{m,n}=\{0\}$ for all $m\leq N_n$.
\end{itemize}
\end{prop}
\begin{rem}~
	\begin{itemize}
		\item[(i)] Statement $(ii)$ of the proposition is crucial since it guarantees that the Sugawara construction gives a well defined representation of the Virasoro algebra on each module. In addition, the same property is required of a strongly graded module, or generalized module, of a (strongly graded) conformal vertex algebra as defined in \cite{HLZ}.
		\item[(ii)] These properties imply that, with the sole exception of $D^{k,0}_0\equiv D^k_0$, none of the modules under consideration belong to the category $\mathscr{O}_{k+2}$. It follows that the results of Kazhdan \& Lusztig \cite{KaLu} cannot be applied to the $SL(2,\R)$ WZW model.
	\end{itemize}
\end{rem}

It is useful to illustrate the structure of the modules by weight diagrams. The weight diagram of ${D}^k_0$ takes the form:

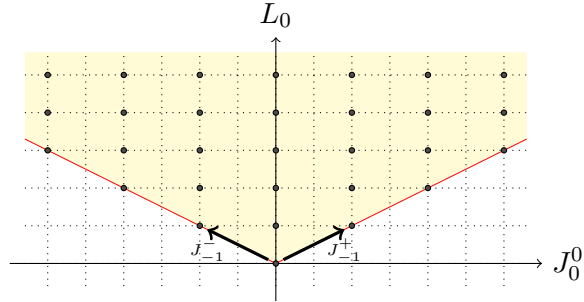
\begin{center}
\begin{tikzpicture}
	\fill[yellow!20!white]	(-3.3,2.8) -- (-3.3,1.65) -- (0,0) -- (3.3,1.65) -- (3.3,2.8) -- (-3.3,2.8);
	\draw[step=.5, dotted] 	(-3.3,-.3) grid (3.3,2.8);
	\draw[->]	(-3.5,0) -- (3.5,0) node[right] {$J^0_0$};
	\draw[->]	(0,-0.5) -- (0,3) node[above] {$L_0$};
	\draw[red,very thin]	(0,0)--(3.3,1.65);
	\draw[red,very thin]	(0,0)--(-3.3,1.65);
	\foreach	\y	in	{0,0.5,...,2.5}
		\filldraw[fill=white!30!black]	(0,\y)	circle	(1pt);
	\foreach	\y	in	{0.5,1,...,2.5}
		{\filldraw[fill=white!30!black]	(-1,\y)	circle	(1pt);
		\filldraw[fill=white!30!black]	(1,\y)		circle	(1pt);
		}
	\foreach	\y	in	{1,1.5,...,2.5}
		{\filldraw[fill=white!30!black]	(-2,\y)	circle	(1pt);
		\filldraw[fill=white!30!black]	(2,\y)		circle	(1pt);
		}
	\foreach	\y	in	{1.5,2,...,2.5}
		{\filldraw[fill=white!30!black]	(-3,\y)	circle	(1pt);
		\filldraw[fill=white!30!black]	(3,\y)		circle	(1pt);
		}
	\draw[->,very thick]	(0.1,0.05) -- (0.9,0.45) node[below=-1pt] {\tiny$J^+_{-1}$};
	\draw[->,very thick]	(-0.1,0.05) -- (-0.9,0.45) node[below=-1pt] {\tiny$J^-_{-1}$};
\end{tikzpicture}
\end{center}
The weights in the module are located in the shaded wedge bounded by the diagonal lines generated by repeated action of $J^-_{-1}$ respectively $J^+_{-1}$ on the highest weight. The weight diagrams of ${D}^{\pm,k}_j$ are illustrated as:
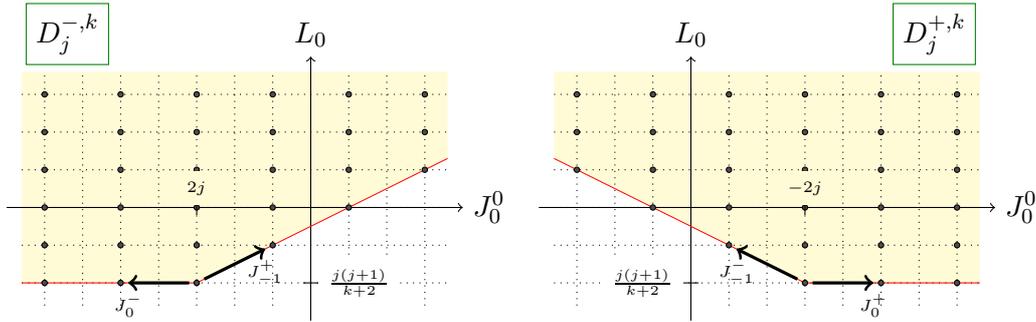
\begin{center}
\begin{tikzpicture}
	\fill[yellow!20!white] (-6.3,0) -- (-4,0) -- (-0.7,1.65) -- (-0.7,2.8) -- (-6.3,2.8) -- (-6.3,0);
	\draw[step=0.5,dotted]	(-6.3,-.3) grid (-0.7,2.8);
	\draw[->]	(-6.5,1)--(-0.5,1) node[right] {$J_0^0$};
	\draw[->]	(-2.5,-.5)--(-2.5,3) node[above] {$L_0$};
	\draw[red,very thin]	(-6.3,0)--(-4,0)--(-0.7,1.65);
	\foreach	\x	in	{-6,...,-4}
		\filldraw[fill=white!30!black]	(\x,0)	circle	(1pt);
	\foreach	\x	in	{-6,...,-3}
		\filldraw[fill=white!30!black]	(\x,0.5)	circle	(1pt);
	\foreach	\x	in	{-6,...,-2}
		\filldraw[fill=white!30!black]	(\x,1)	circle	(1pt);
	\foreach	\x	in	{-6,...,-1}{
		\filldraw[fill=white!30!black]	(\x,1.5)	circle	(1pt);
		\filldraw[fill=white!30!black]	(\x,2)	circle	(1pt);
		\filldraw[fill=white!30!black]	(\x,2.5)	circle	(1pt);
		}
	\draw[->,very thick]	(-3.9,0.05)--(-3.1,.45) node[below] {\tiny $J^+_{-1}$};
	\draw[->,very thick]	(-4.1,0)--(-4.9,0) node[below] {\tiny $J^-_0$};
	\draw	(-4,.9)--(-4,1.1) node[above=-2pt,fill=yellow!20!white] {\tiny $2j$};
	\draw	(-2.6,0)--(-2.4,0) node[right,fill=white] {\tiny$\frac{j(j+1)}{k+2}$};
	\draw	(-5.7,3.3)	node[draw=green!50!black] {${D}^{-,k}_j$};
%%%%%%%%%%%%%%%%%%%%%%%%%%%%%%%
	\fill[yellow!20!white] (6.3,0) -- (4,0) -- (0.7,1.65) -- (0.7,2.8) -- (6.3,2.8) -- (6.3,0);
	\draw[step=0.5,dotted]	(.7,-0.3) grid (6.3,2.8);
	\draw[->]	(.5,1)--(6.5,1) node[right] {$J^0_0$};
	\draw[->]	(2.5,-0.5)--(2.5,3) node[above] {$L_0$};
	\draw[red,very thin]	(6.3,0) -- (4,0) -- (0.7,1.65);
	\foreach	\x	in	{4,...,6}
		\filldraw[fill=white!30!black]	(\x,0)	circle	(1pt);
	\foreach	\x	in	{3,...,6}
		\filldraw[fill=white!30!black]	(\x,0.5)	circle	(1pt);
	\foreach	\x	in	{2,...,6}
		\filldraw[fill=white!30!black]	(\x,1)	circle	(1pt);
	\foreach	\x	in	{1,...,6}{
		\filldraw[fill=white!30!black]	(\x,1.5)	circle	(1pt);
		\filldraw[fill=white!30!black]	(\x,2)	circle	(1pt);
		\filldraw[fill=white!30!black]	(\x,2.5)	circle	(1pt);
		}
	\draw[->,very thick]	(3.9,0.05)--(3.1,.45) node[below] {\tiny $J^-_{-1}$};
	\draw[->,very thick]	(4.1,0)--(4.9,0) node[below] {\tiny $J^+_0$};
	\draw	(4,.9)--(4,1.1) node[above=-2pt,fill=yellow!20!white] {\tiny $-2j$};
	\draw	(2.6,0)--(2.4,0) node[left,fill=white] {\tiny$\frac{j(j+1)}{k+2}$};
	\draw	(5.7,3.3)	node[draw=green!50!black] {${D}^{+,k}_j$};
\end{tikzpicture}
\end{center}

\noindent
The weight diagrams for ${D}^{k,\pm 1}_0$ are illuminating in comparison with those of $D^{\pm,k}_j$:

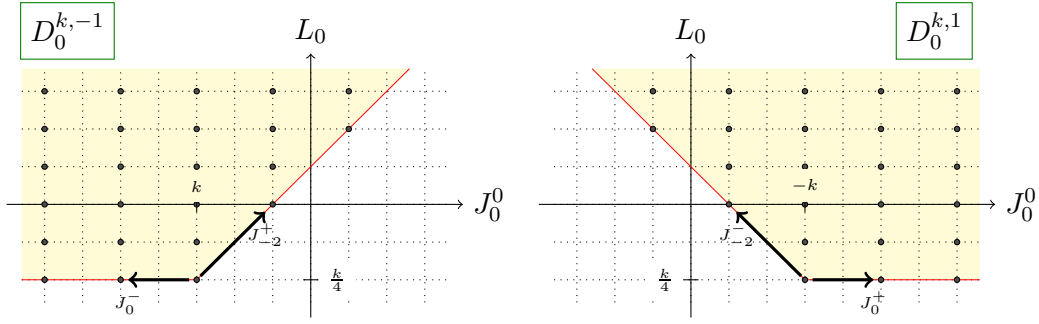
\begin{center}
\begin{tikzpicture}
	\fill[yellow!20!white] (-6.3,0) -- (-4,0) -- (-1.2,2.8)  -- (-6.3,2.8) -- (-6.3,0);
	\draw[step=0.5,dotted]	(-6.3,-.3) grid (-0.7,2.8);
	\draw[->]	(-6.5,1)--(-0.5,1) node[right] {$J_0^0$};
	\draw[->]	(-2.5,-.5)--(-2.5,3) node[above] {$L_0$};
	\draw[red,very thin]	(-6.3,0)--(-4,0)--(-1.2,2.8);
	\foreach	\x	in	{-6,...,-4}
		\filldraw[fill=white!30!black]	(\x,0)	circle	(1pt);
	\foreach	\x	in	{-6,...,-4}
		\filldraw[fill=white!30!black]	(\x,0.5)	circle	(1pt);
	\foreach	\x	in	{-6,...,-3}{
		\filldraw[fill=white!30!black]	(\x,1)	circle	(1pt);
		\filldraw[fill=white!30!black]	(\x,1.5)	circle	(1pt);
		}
	\foreach	\x	in	{-6,...,-2}{
		\filldraw[fill=white!30!black]	(\x,2)	circle	(1pt);
		\filldraw[fill=white!30!black]	(\x,2.5)	circle	(1pt);
		}
	\draw[->,very thick]	(-3.95,0.05)--(-3.1,.9) node[below] {\tiny $J^+_{-2}$};
	\draw[->,very thick]	(-4.1,0)--(-4.9,0) node[below] {\tiny $J^-_0$};
	\draw	(-4,.9)--(-4,1.1) node[above=-2pt,fill=yellow!20!white] {\tiny $k$};
	\draw	(-2.6,0)--(-2.4,0) node[right,fill=white] {\tiny$\frac{k}{4}$};
	\draw	(-5.7,3.3)	node[draw=green!50!black] {${D}^{k,-1}_0$};
%%%%%%%%%%%%%%%%%%%%%%%%%%%%%%%
	\fill[yellow!20!white] (6.3,0) -- (4,0) -- (1.2,2.8) -- (6.3,2.8) -- (6.3,0);
	\draw[step=0.5,dotted]	(.7,-0.3) grid (6.3,2.8);
	\draw[->]	(.5,1)--(6.5,1) node[right] {$J^0_0$};
	\draw[->]	(2.5,-0.5)--(2.5,3) node[above] {$L_0$};
	\draw[red,very thin]	(6.3,0) -- (4,0) -- (1.2,2.8);
	\foreach	\x	in	{4,...,6}
		\filldraw[fill=white!30!black]	(\x,0)	circle	(1pt);
	\foreach	\x	in	{4,...,6}
		\filldraw[fill=white!30!black]	(\x,0.5)	circle	(1pt);
	\foreach	\x	in	{3,...,6}{
		\filldraw[fill=white!30!black]	(\x,1)	circle	(1pt);
		\filldraw[fill=white!30!black]	(\x,1.5)	circle	(1pt);
		}
	\foreach	\x	in	{2,...,6}{
		\filldraw[fill=white!30!black]	(\x,2)	circle	(1pt);
		\filldraw[fill=white!30!black]	(\x,2.5)	circle	(1pt);
		}
	\draw[->,very thick]	(3.95,0.05)--(3.1,.9) node[below] {\tiny $J^-_{-2}$};
	\draw[->,very thick]	(4.1,0)--(4.9,0) node[below] {\tiny $J^+_0$};
	\draw	(4,.9)--(4,1.1) node[above=-2pt,fill=yellow!20!white] {\tiny $-k$};
	\draw	(2.6,0)--(2.4,0) node[left,fill=white] {\tiny$\frac{k}{4}$};
	\draw	(5.7,3.3)	node[draw=green!50!black] {${D}^{k,1}_0$};
\end{tikzpicture}
\end{center}
Note that the extremal weights (the tips of the wedges) of $D^{k,\pm1}_0$, i.e. $(k/4,\mp k)$, coincide with the extremal (highest resp. lowest) weights of $D^{\pm,k}_{k/2}$, which are both reducible. In fact, we have
\begin{prop}\label{prop:scquotient}
	$D^{k,\pm 1}_0\cong D^{\pm,k}_{k/2}/W^\pm_{k/2}$, where $W^\pm_{k/2}$ are the submodules generated from $J^\mp_{-1}|k/2,0\rangle$
\end{prop}
\begin{proof}
	Since $D^{k,\pm 1}_0$ are simple highest/lowest weight modules they are quotients of the corresponding Verma modules, $D^{\pm,k}_{k/2}$, by the maximal submodules. It is easily checked that $J^\mp_{-1}|k/2,0\rangle$ are highest/lowest weight vectors, and since they reside at Virasoro level 1 and at the boundary of the weight wedge, the submodules generated from these vectors are maximal submodules.
\end{proof}

\noindent
For completeness, the weight diagrams for ${D}^{k,s}_0$, and ${D}^{+,k,s}_j$ for $s\neq 0,\pm 1$,  are illustrated below.

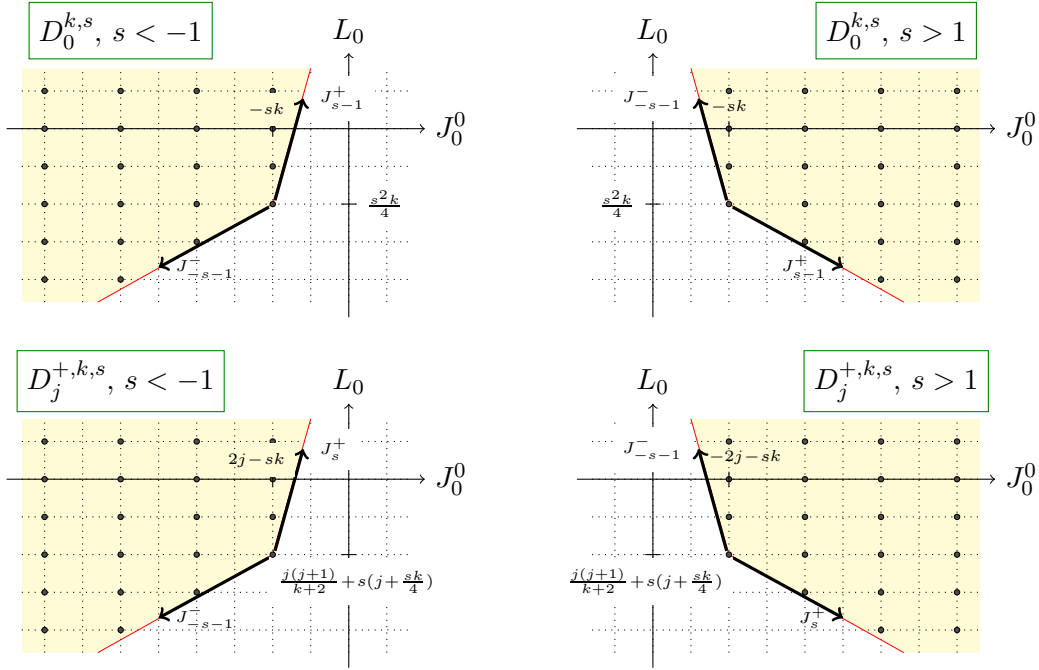
\begin{center}
\begin{tikzpicture}
	\fill[yellow!20!white] (-6.3,-.3) -- (-5.3,-.3) -- (-3,1) -- (-2.5,2.8)  -- (-6.3,2.8) -- (-6.3,0);
	\draw[step=0.5,dotted]	(-6.3,-.3) grid (-1.2,2.8);
	\draw[->]	(-6.5,2)--(-1,2) node[right] {$J_0^0$};
	\draw[->]	(-2,-.5)--(-2,3) node[above] {$L_0$};
	\draw[red,very thin]	(-5.3,-0.3)--(-3,1)--(-2.5,2.8);
	\foreach	\x	in	{-6,...,-3}{
		\filldraw[fill=white!30!black]	(\x,1)	circle	(1pt);
		\filldraw[fill=white!30!black]	(\x,1.5)	circle	(1pt);
		\filldraw[fill=white!30!black]	(\x,2)	circle	(1pt);
		\filldraw[fill=white!30!black]	(\x,2.5)	circle	(1pt);
		}
	\foreach	\x	in	{-6,...,-4}{
		\filldraw[fill=white!30!black]	(\x,0.5)	circle	(1pt);
		}
	\foreach	\x	in	{-6,...,-5}{
		\filldraw[fill=white!30!black]	(\x,0)	circle	(1pt);
		}
	\draw[->,very thick]	(-2.97,1.05)--(-2.6,2.4) node[right=2pt,fill=white] {\tiny $J^+_{s-1}$};
	\draw[->,very thick]	(-3.05,0.96)--(-4.5,0.16) node[right=2pt] {\tiny $J^-_{-s-1}$};
	\draw	(-3,1.9)--(-3,2.1) node[above=-2pt,fill=yellow!20!white] {\tiny\!\!\!\!\!\! $-sk$};
	\draw	(-2.1,1)--(-1.9,1) node[right,fill=white] {\tiny$\frac{s^2k}{4}$};
	\draw	(-5,3.3)	node[draw=green!50!black] {${D}^{k,s}_0$, $s<-1$};
%%%%%%%%%%%%%%%%%%%%%%%%%%%%%%%
	\fill[yellow!20!white] (6.3,-.3) -- (5.3,-0.3) -- (3,1) -- (2.5,2.8) -- (6.3,2.8) -- (6.3,-0.3);
	\draw[step=0.5,dotted]	(1.2,-0.3) grid (6.3,2.8);
	\draw[->]	(1,2)--(6.5,2) node[right] {$J^0_0$};
	\draw[->]	(2,-0.5)--(2,3) node[above] {$L_0$};
	\draw[red,very thin]	(5.3,-0.3) -- (3,1) -- (2.5,2.8);
	\foreach	\x	in	{5,...,6}
		\filldraw[fill=white!30!black]	(\x,0)	circle	(1pt);
	\foreach	\x	in	{4,...,6}
		\filldraw[fill=white!30!black]	(\x,0.5)	circle	(1pt);
	\foreach	\x	in	{3,...,6}{
		\filldraw[fill=white!30!black]	(\x,1)	circle	(1pt);
		\filldraw[fill=white!30!black]	(\x,1.5)	circle	(1pt);
		}
	\foreach	\x	in	{3,...,6}{
		\filldraw[fill=white!30!black]	(\x,2)	circle	(1pt);
		\filldraw[fill=white!30!black]	(\x,2.5)	circle	(1pt);
		}
	\draw[->,very thick]	(2.97,1.05)--(2.6,2.4) node[left=2pt,fill=white] {\tiny $J^-_{-s-1}$};
	\draw[->,very thick]	(3.05,0.96)--(4.5,0.16) node[left=2pt] {\tiny $J^+_{s-1}$};
	\draw	(3,1.9)--(3,2.1) node[above=-1pt] {\tiny $-sk$};
	\draw	(2.1,1)--(1.9,1) node[left,fill=white] {\tiny$\frac{s^2k}{4}$};
	\draw	(5.2,3.3)	node[draw=green!50!black] {${D}^{k,s}_0$, $s>1$};
\end{tikzpicture}
\end{center}

\begin{center}
\begin{tikzpicture}
	\fill[yellow!20!white] (-6.3,-.3) -- (-5.3,-.3) -- (-3,1) -- (-2.5,2.8)  -- (-6.3,2.8) -- (-6.3,0);
	\draw[step=0.5,dotted]	(-6.3,-.3) grid (-1.2,2.8);
	\draw[->]	(-6.5,2)--(-1,2) node[right] {$J_0^0$};
	\draw[->]	(-2,-.5)--(-2,3) node[above] {$L_0$};
	\draw[red,very thin]	(-5.3,-0.3)--(-3,1)--(-2.5,2.8);
	\foreach	\x	in	{-6,...,-3}{
		\filldraw[fill=white!30!black]	(\x,1)	circle	(1pt);
		\filldraw[fill=white!30!black]	(\x,1.5)	circle	(1pt);
		\filldraw[fill=white!30!black]	(\x,2)	circle	(1pt);
		\filldraw[fill=white!30!black]	(\x,2.5)	circle	(1pt);
		}
	\foreach	\x	in	{-6,...,-4}{
		\filldraw[fill=white!30!black]	(\x,0.5)	circle	(1pt);
		}
	\foreach	\x	in	{-6,...,-5}{
		\filldraw[fill=white!30!black]	(\x,0)	circle	(1pt);
		}
	\draw[->,very thick]	(-2.97,1.05)--(-2.6,2.4) node[right=2pt,fill=white] {\tiny $J^+_{s}$};
	\draw[->,very thick]	(-3.05,0.96)--(-4.5,0.16) node[right=2pt] {\tiny $J^-_{-s-1}$};
	\draw	(-3,1.9)--(-3,2.1) node[above=-2pt,fill=yellow!20!white] {\tiny\!\!\!\!\!\!\!\!\!\!\!\! $2j\!-\!sk$};
	\draw	(-2.1,1)--(-1.9,1) node[below=1pt,fill=white] {\tiny$\frac{j(j+1)}{k+2}\! +\! s(j\!+\!\frac{sk}{4})$};
	\draw	(-5,3.3)	node[draw=green!50!black] {${D}^{+,k,s}_j$, $s<-1$};
%%%%%%%%%%%%%%%%%%%%%%%%%%%%%%%
	\fill[yellow!20!white] (6.3,-.3) -- (5.3,-0.3) -- (3,1) -- (2.5,2.8) -- (6.3,2.8) -- (6.3,-0.3);
	\draw[step=0.5,dotted]	(1.2,-0.3) grid (6.3,2.8);
	\draw[->]	(1,2)--(6.5,2) node[right] {$J^0_0$};
	\draw[->]	(2,-0.5)--(2,3) node[above] {$L_0$};
	\draw[red,very thin]	(5.3,-0.3) -- (3,1) -- (2.5,2.8);
	\foreach	\x	in	{5,...,6}
		\filldraw[fill=white!30!black]	(\x,0)	circle	(1pt);
	\foreach	\x	in	{4,...,6}
		\filldraw[fill=white!30!black]	(\x,0.5)	circle	(1pt);
	\foreach	\x	in	{3,...,6}{
		\filldraw[fill=white!30!black]	(\x,1)	circle	(1pt);
		\filldraw[fill=white!30!black]	(\x,1.5)	circle	(1pt);
		}
	\foreach	\x	in	{3,...,6}{
		\filldraw[fill=white!30!black]	(\x,2)	circle	(1pt);
		\filldraw[fill=white!30!black]	(\x,2.5)	circle	(1pt);
		}
	\draw[->,very thick]	(2.97,1.05)--(2.6,2.4) node[left=2pt,fill=white] {\tiny $J^-_{-s-1}$};
	\draw[->,very thick]	(3.05,0.96)--(4.5,0.16) node[left=2pt] {\tiny $J^+_{s}$};
	\draw	(3,1.9)--(3,2.1) node[above=-1pt] {\tiny\quad\ \  $-2j\!-\!sk$};
	\draw	(2.1,1)--(1.9,1) node[below=1pt,fill=white] {\tiny$\frac{j(j+1)}{k+2}\! +\! s(j\!+\!\frac{sk}{4})$};
	\draw	(5.2,3.3)	node[draw=green!50!black] {${D}^{+,k,s}_j$, $s>1$};
\end{tikzpicture}
\end{center}
The corresponding diagrams for ${D}^{-,k,s}_j$ are obtained via Proposition \ref{prop:sfequiv}. It is obvious from these diagrams that the homogeneous eigenspaces become trivial when moving down far enough along constant $J^0_0$ eigenvalue.\\

\noindent
It is of interest to determine how the modules ${D}^{\pm,k,s}_j$ decompose as $\mathcal{U}(\mathfrak{sl}_2)$--modules with respect to various $\mathfrak{sl}_2$ subalgebras of $\widehat{\mathfrak{sl}}_2$. In full generality this is outside the scope of this article, but the following result will be useful in the next section.

\begin{prop}\label{prop:indec}
	No module of the type ${D}^{k,s}_0$ or ${D}^{\pm,k,s}_j$ with $j>k/2$ decompose in a direct sum of unitarizable $\mathcal{U}(\mathfrak{sl}_2)$--modules w.r.t. $\widetilde{\omega}$, where $\mathfrak{sl}_2$ is the horizontal subalgebra. Furthermore, with the exception of ${D}^k_0$ these modules are reducible but not fully decomposable as modules over the horizontal $\mathcal{U}(\mathfrak{sl}_2)$.
\end{prop}
\begin{proof}
	The first statement is obviously true for ${D}^k_0$, since, as a $\mathcal{U}(\mathfrak{sl}_2)$--module it is decomposable in a direct sum of finite dimensional submodules, not all of which are one dimensional. Consider next ${D}^{+,k}_j$. Note that all vectors of the type $\left(J^-_{-1}\right)^n|j,0\rangle$ are lowest weight vectors w.r.t. the horizontal subalgebra, but for $n>-j$ this has negative $J_0^0$ eigenvalue. The corresponding lowest weight module therefore does not belong to the class of unitarizable modules. Inspection of the weight diagrams of the rest of the modules reveals that moving far enough upwards along the boundary one will always encounter a change in sign of the $J_0^0$ eigenvalues of the highest or lowest weight vectors along the boundary, and the same argument thus shows the first statement for all of the remaining cases.
The second statement will be shown for ${D}^{+,k}_j$, the proof is completely analogous for the other modules. Choose an integer $n$ such that $j+n>0$ and consider the $\mathcal{U}(\mathfrak{sl}_2)$--module freely generated from the lowest weight vector $\omega^{(n)}\defas (J^-_{-1})^n|j,0\rangle$, satisfying $J^0_0\omega^{(n)}=-2(j+n)\omega^{(n)}$. This module, with a basis given by the vectors $(J^+_0)^m\omega^{(n)}$, is thus equivalent to the Verma module with lowest weight $-2(j+n)<0$. It immediately follows that it is not fully decomposable, with a maximal submodule generated from the lowest weight vector $(J^+_0)^{2(j+n)+1}\omega^{(n)}$. Further analysis shows that the horizontal submodules in general do not decompose as a direct sum of Verma modules, the structure of embedded submodules is more involved.
\end{proof}

%%%%%%%%%%%%%%%%%%%%%%%%%%%%%%%%%%%%%%%%%%%%
\subsection{A contragredient of \texorpdfstring{$\mathcal{U}(\widehat{\mathfrak{sl}}_2)$}{U(affine sl2)}--modules}\label{sec:cont}
Recall that for an integrable module $M$ of $\mathcal{U}(\widehat{\mathfrak{g}})$ one defines its contragredient, $M^+$, by using the $\N_0$ grading given by $L_0$, $M=\oplus_{n\in\N_0}M_n$. As a vector space $M^+$ is the restricted dual of $M$, $M^+=\oplus_{n\in\N_0}M_n^*$, and the action of $X_m$ is given by minus the transpose of $X_{-m}$. This definition works well for integrable modules since then each homogeneous subspace $M_n$ is finite dimensional. Of the modules considered in this article, only ${D}^{k}_0$ is integrable, and it seems unlikely that the same notion of contragredient will be useful.

One conceivable modification would be to equip each subspace $M_n$ with the structure of a Hilbert space, and defining a ``restricted conjugate'' representation. The obvious first guess would then be to let $M_0$ be a unitary representation of the horizontal subalgebra, inducing an inner product on the remaining levels by the Shapovalov form (or the analogue with the involutive anti automorphism $\omega$ rather than the Cartan involution), and completing each homogeneous subspace with respect to the inner product norm. The resulting sesquilinear form, however, is not positive definite in ${D}^{\pm,k,s}_j$ (hence also not in $ \uD^{\pm,k,s}_j$), as follows from the proofs of the no ghost theorems in \cite{Hwang,MaOg}, and does not give an inner product. Note that this holds both with respect to $\omega$ and to the Cartan involution (corresponding to the affine version of $su(2)$), as can be read off from the Kac--Kazhdan determinant. Next one could try introducing by hand inner products on each $M_n$ such that the horizontal subalgebra is represented unitarily. The first statement of Proposition \ref{prop:indec} implies that the adjoints of the generators $J^0_0$, $J^\pm_0$ would have to be different in different submodules of $M_n$. The situation, however, is worse than that. The second statement of Proposition \ref{prop:indec} implies that some subspaces $M_n$ cannot be equipped with inner products such that the horizontal subalgebra acts unitarily, since such structures cannot exist on reducible but indecomposable modules.

It might still be possible to introduce a topology on the homogeneous subspaces and define a contragredient on the ``restricted continuous dual'' of $M$. This possibility will not be investigated here. We will instead use the full $\Z\times\Z$ grading to define a contragredient. 
\begin{defn}\label{def:cont}
	Let $V$ denote any module of the type ${D}^{k,s}_0$, or ${D}^{\pm,k,s}_j$, for $s\in\Z$ and $j<0$, and denote by $V_{m,n}$ the homogeneous subspace of grade $(m,n)\in\Z\times\Z$. The contragredient of $V$, denoted $V^\cont$, is defined as the vector space
	\begin{equation}
		V^\cont \defas \bigoplus_{(m,n)\in\Z\times\Z}V_{m,n}^*,
	\end{equation}
	with action defined by
	\begin{align}
		J^a_n\varphi(v) &\defas - \varphi(J^a_{-n}v),\ a\in\{0,\pm\}\\
		K\varphi(v) &\defas \varphi(Kv)\\
		d\varphi(v) &\defas \varphi(dv),
	\end{align}
	and extended by linearity, for any $\varphi\in V^\cont$ and $v\in V$. In the formulas defining the action it is assumed that $V_{m,n}^*$ is extended trivially to all $V_{k,l}$ for $(m,n)\neq(k,l)$.
\end{defn}
Note that, applied to integrable modules Definition \ref{def:cont} becomes equivalent to the conventional definition.
That this definition is sensible follows from
\begin{thm}\label{thm:cont}
	For any $k\in-\N_{>2}$, $j\in\{\frac{k+1}{2},\ldots,-\frac{1}{2}\}$, $s\in\Z$ the following equivalences hold
	\begin{align}
		\big({D}^{k,s}_0\big)^\cont &\cong {D}^{k,-s}_0\\
		\big({D}^{\pm,k,s}_j\big)^\cont &\cong {D}^{\mp,k,-s}_j
	\end{align}
\end{thm}
The following lemma is useful in proving the theorem.
\begin{lem}\label{lem:sfcont}
	For any $\mathcal{U}(\widehat{\mathfrak{sl}}_2)$--module $V$ we have the \emph{equality}
	\[ \left( V^s\right)^\cont = \left(V^\cont\right)^{-s},\]
	where the superscript $s$ denotes the action of spectral flow with $s\in\Z$.
\end{lem}
\begin{proof}
	Note that the underlying vector spaces of $\left( V^s\right)^\cont$ and $\left(V^\cont\right)^{-s}$ are the same, so it suffices to check the action. Definitions \ref{def:sflow} and \ref{def:cont} immediately give
	\begin{align}
		\left. \begin{array}{ll}
				J^\pm_n\varphi & = - \varphi\circ J^\pm_{-n\mp s}\\
				J^-_n\varphi & = -\varphi\circ(J^0_{-n} - sK\delta_n)\\
				L_0\varphi &= \varphi\circ(L_0-\frac{s}{2}J^0_0+\frac{s^2}{4}K)
			\end{array}
			\right\} & \text{ on } \left( V^s\right)^\cont \\
		\left. \begin{array}{ll}
				J^\pm_n\varphi & = - \varphi\circ J^\pm_{-n\pm s}\\
				J^-_n\varphi & = -\varphi\circ(J^0_{-n} +sK\delta_n)\\
				L_0\varphi &= \varphi\circ(L_0 +\frac{s}{2}J^0_0+\frac{s^2}{4}K)
			\end{array}
			\right\} & \text{ on } \left( V^\cont\right)^s,
	\end{align}
	where on the right hand side representation morphisms of $V$ have been left out. Comparison immediately gives the proof.
\end{proof}
The means to prove the theorem are now at hand.
\begin{proof}[Proof of Theorem \ref{thm:cont}]
Consider first the equivalence $\big({D}^{-,k}_j\big)^\cont\cong{D}^{+,k}_j$, and fix $\varphi\in\big({D}^{-,k}_j\big)_{0,0}^*$ by $\varphi(|j,0\rangle)=1$. It follows immediately from Definition \ref{def:cont} that
\begin{align}
	J^0_0\varphi &= -2j\varphi & L_0\varphi &= \frac{j(j+1)}{k+2}\varphi\nonumber\\
	J^-_0\varphi &= 0 & J^a_n\varphi &= 0,\ n\in\N, a\in\{0,\pm\},\nonumber
\end{align}
i.e. $\varphi$ is a highest weight vector with the same weight as the highest weight vector of ${D}^{+,k}_j$. The submodule generated from $\varphi$ is then equivalent to a quotient module of the corresponding Verma module, but since the latter is simple it follows that the submodule generated from $\varphi$ is equivalent to ${D}^{+,k}_j$. Under this inclusion it follows that $({D}^{+,k}_j)_{n,m}\subset\big({D}^{-,k}_j\big)^*_{n,-m}$, and since all homogeneous subspaces are finite dimensional we know that
\[\mathrm{dim}\Big( \big({D}^{-,k}_j\big)_{n,-m}^*\Big) = \mathrm{dim}\Big(({D}^{-,k}_j)_{n,-m}\Big). \]
Consider Poincar\'e--Birkhoff--Witt bases of ${D}^{\pm,k}_j$, labeled by triples $\{(\vec{k}^+,\vec{k}^0,\vec{k}^-)\}$ where each $\vec{k}^a$ is a $\N_0$ sequence with almost all entries $=0$.  In the case of ${D}^{\pm,k}_j$ the sequences $\vec{k}^0$ and $\vec{k}^\mp$ are required to have vanishing first entries, $k^0_0=0$, $k^\mp_0=0$. Using the (abuse of) notation $d^\pm_{nm}=\mathrm{dim}\Big(({D}^{\pm,k}_j)_{n,m}\Big)$, the dimensions of the homogeneous subspaces can now be expressed as
\begin{align}
	d^+_{nm} &= \Bigg| \left\{ (\vec{k}^+,\vec{k}^0,\vec{k}^-)\bigg| \sum_{i\in\N_0} i(k^+_i+k^0_i+k^-_i)=n, \sum_{i\in\N_0} (k^+_i-k^-_i) = 2m \right\} \Bigg|\label{dim1}\\
	d^-_{nm} &= \Bigg| \left\{ (\vec{k}^+,\vec{k}^0,\vec{k}^-)\bigg| \sum_{i\in\N_0} i(k^+_i+k^0_i+k^-_i)=n, \sum_{i\in\N_0} (k^-_i-k^+_i) = -2m \right\} \Bigg|.\label{dim2}
\end{align}
Note that \eqref{dim1} and \eqref{dim2} are not identical due to the different constraints on the sequences $\vec{k}^\pm$. It immediately follows that the map
\[ (\vec{k}^+,\vec{k}^0,\vec{k}^-)\mapsto (\vec{k}^-,\vec{k}^0,\vec{k}^+)\] induces a linear isomorphism $({D}^{+,k}_j)_{n,m}\stackrel{\sim}{\rightarrow}({D}^{-,k}_j)_{n,-m}$. Since the dimensions of every homogeneous subspace of the respective modules coincide we have shown that the submodule generated from $\varphi$ spans $\big({D}^{-,k}_j\big)^\cont$, and thus $\big({D}^{-,k}_j\big)^\cont\cong{D}^{+,k}_j$. The equivalence $\big({D}^{+,k}_j\big)^\cont\cong{D}^{-,k}_j$ follows immediately by applying the contragredient to $\big({D}^{-,k}_j\big)^\cont\cong{D}^{+,k}_j$ since $\left((V)^\cont\right)^\cont\cong V$ (this is a trivial consequence of the definition of contragredient).

As stated above, the definition of contragredient given here coincides with the conventional definition on integrable modules, so it follows trivially that $\left({D}^k_0\right)^\cont\cong{D}^k_0$.

The remaining equivalences now follow from the ones already shown together with Lemma \ref{lem:sfcont}.
\end{proof}

\noindent
Having defined a contragredient of (a class of) $\ena$--modules one would like to know what structure $\mathcal{C}^k$ inherits from this. Consider therefore the two--point conformal blocks on $\C\mathbb{P}^1$.
\begin{thm}\label{thm:2ptblocks}
	Let $0,\infty\in \C\mathbb{P}^1$ be equipped with the standard quasi global coordinates $z$ and $w=1/z$. We then have
	\[\mathcal{H}\left((0,\infty),\left(D^{\pm,k,s}_i,D^{\mp,k,-t}_j\right)\right)\cong\delta_{s,t}\delta_{i,j}\C.\]
\end{thm}

Elements in a space of two--point conformal blocks on $\C\mathbb{P}^1$ are known as Ishibashi states, and there is a standard way to obtain them \cite{Ishibashi}. The line of proof here will largely follow that construction, with slight modifications. Most importantly, as formulated here there is no need for an inner product on the modules.
It is furthermore likely that for $s=t=0$, the result is possible to extract from the Ishibashi states constructed in \cite{Hikida}, although the formalism differs significantly.
The case $s\neq t$ requires results from the next section, and is therefore is relegated to appendix \ref{sec:contproof}. We focus here on the case $s=t$.
\begin{proof} ($s=t$)
Note first that $\mathfrak{sl}_2(\C\mathbb{P}^1-0-\infty)\cong \mathfrak{g}[z,z^{-1}]$, and the action of this algebra amounts to the action of elements $J^a_n\oti\mathbf{1}+\mathbf{1}\oti J^a_{-n}$ for $a=0,\pm$, $n\in\Z$. Furthermore, use Theorem \ref{thm:cont} and replace $D^{\mp,k,-t}_j$ with $\big(D^{\pm,k,t}_j\big)^\cont$.
	Consider first the case of ${D}^{+,k}_j$ with $i=j$, $s=t=0$. A coinvariant is then an element $\beta\in\left({D}^{+,k}_j\oti\big({D}^{+,k}_j\big)^\cont\right)^*$ satisfying
	\begin{equation}
		\beta\circ(J^a_n\oti\mathbf{1}+\mathbf{1}\oti J^a_{-n})=0.
		\label{eq:coinvcond}
	\end{equation}
	Let $\{|\vec{p}\rangle\}$ denote a basis of ${D}^{+,k}_j$, adapted to the weight decomposition, and let $\{\varphi_{\vec{p}}\}$ denote the dual basis of the contragredient module. Note that ${D}^{+,k}_j$ can be identified with a subspace of the space of linear functionals on $\big({D}^{+,k}_j\big)^\cont$, and use this identification to express $\beta$ as
	\begin{equation}
		\beta=\sum_{\vec{p},\vec{q}}c_{\vec{p},\vec{q}}\varphi_{\vec{p}}\oti|\vec{q}\rangle.
	\label{eq:coinvansatz}
	\end{equation}
	With no restrictions on $c_{\vec{p},\vec{q}}\in\C$ it is clear that any element in the algebraic dual can be represented in this way. Using the Sugawara form of $L_0$ it is easily shown that the coinvariance conditions implies
	\begin{equation}
		\beta\circ(L_0\oti\mathbf{1}-\mathbf{1}\oti L_0)=0
	\label{eq:vircoinv}
	\end{equation}
	(note the minus sign). Assume $\beta$ is a nonzero coinvariant. Then there exists at least one basis element $|\vec{r}\rangle\oti\varphi_{\vec{s}}$ such that $\beta(|\vec{r}\rangle\oti\varphi_{\vec{s}})\neq 0$. The coinvariance conditions \eqref{eq:coinvcond} and \eqref{eq:vircoinv} imply that if $|\vec{r}\rangle$ has grade $(n,m)$ then $\varphi_{\vec{s}}$ has grade $(n,-m)$. Using the fact that both modules are Verma modules, and expressing the bases in PBW bases, the relation \eqref{eq:coinvcond} can be used to relate $\beta(|\vec{r}\rangle\oti\varphi_{\vec{s}})$ to $\beta(|0\rangle\oti\varphi_0\rangle$, resulting in
	\begin{equation}
		0\neq \beta(|\vec{r}\rangle\oti\varphi_{\vec{s}})=\alpha_{\vec{r},\vec{s}}\beta(|0\rangle\oti\varphi_0),
	\label{eq:coinvhw}
	\end{equation}
	where $|0\rangle$ respectively $\varphi_0$ denote the highest/lowest weight vectors, for some $\alpha_{\vec{r},\vec{s}}\in\C^\times$. Clearly, for any element $v$ such that $\beta(v)\neq 0$, equation \eqref{eq:coinvhw} has to be satisfied with some $\alpha_v\in\C^\times$. For two nonzero solutions  of \eqref{eq:coinvcond}, $\beta$ and $\beta'$ , also $\beta+\lambda \beta'$ is a solution for any $\lambda\in\C$, and there exists a $\lambda$ such that $\beta(|0\rangle\oti\varphi_0)+\lambda\beta'(|0\rangle\oti\varphi_0)=0$. It thus follows that any two solutions to the coinvariance conditions \eqref{eq:coinvcond} are proportional, and the linear space of solutions is at most one--dimensional. To show that it is not zero--dimensional, show that the nonzero
	\[\beta=\sum_{\{\vec{p}\}}\varphi_{\vec{p}}\oti |\vec{p}\rangle\]
	indeed solves \eqref{eq:coinvcond} (note that this ansatz pairs the grades $(n,m)$ and $(n,-m)$ as required). If it did not, there would exist a basis element $|\vec{r}\rangle\oti\varphi_{\vec{s}}$ such that $\beta((J^a_n\oti\mathbf{1}+\mathbf{1}\oti J^a_{-n})|\vec{r}\rangle\oti\varphi_{\vec{s}})\neq 0$ for some $a$ and $n$. This cannot be, since
	\begin{align}
		\beta((J^a_n\oti\mathbf{1}+\mathbf{1}\oti J^a_{-n})|\vec{r}\rangle\oti\varphi_{\vec{s}}) & = \sum_{\{\vec{p}\}}\left(\varphi_{\vec{p}}(J^a_n|\vec{r}\rangle)\varphi_{\vec{s}}(|\vec{p}\rangle)+\varphi_{\vec{p}}(|\vec{r}\rangle)|\vec{p}\rangle(J^a_{-n}\varphi_{\vec{s}})\right)\nonumber\\
		& = \sum_{\{\vec{p}\}}\left(\varphi_{\vec{p}}(J^a_n|\vec{r}\rangle)\varphi_{\vec{s}}(|\vec{p}\rangle)+\varphi_{\vec{p}}(|\vec{r}\rangle)(J^a_{-n}\varphi_{\vec{s}})(|\vec{p}\rangle)\right)\nonumber\\
		& = \sum_{\{\vec{p}\}}\left(\varphi_{\vec{p}}(J^a_n|\vec{r}\rangle)\varphi_{\vec{s}}(|\vec{p}\rangle)-\varphi_{\vec{p}}(|\vec{r}\rangle)\varphi_{\vec{s}}(J^a_{n}|\vec{p}\rangle)\right)\nonumber\\
		& = 0.\nonumber
	\end{align}
	The proposition has now been shown to hold for $i=j$, $s=t=0$. The modifications needed for ${D}^{-,k}_j$ are trivial and will not be spelled out. 

Allowing $i\neq j$ but still keeping $s=t=0$, the conditions \eqref{eq:coinvcond} still imply a relation like \eqref{eq:coinvhw}, with the only difference that $\varphi_0$ is replaced by some other element at Virasoro level $0$	. However, restricting \eqref{eq:coinvcond} to $n=0$ and applying to the ansatz \eqref{eq:coinvansatz} removes any Virasoro level $0$ elements. This follows from Proposition \ref{prop:sl2nonsemi}.
 ~Thus the proof is done for $s=t=0$.
	
Next, consider again $i=j$ but allow $s=t\neq 0$ The only step in the proof that does not immediately hold for ${D}^{\pm,k,s}_j$, $s\neq 0$, is in showing the relation \eqref{eq:coinvhw} where the existence of a PBW basis was used. On ${D}^{\pm,k,s}_j$ one can instead use the basis given by $\vartheta_s$ acting on a PBW basis to arrive at \eqref{eq:coinvhw}, where $|0\rangle$ then represents the cyclic vector from which the whole module is generated.

Allowing $i\neq j$ but keeping $s=t$ it is straightforward to generalize the line of proof from $s=t=0$. The remaining case of $s\neq t$ will be treated in appendix \ref{sec:contproof}.
\end{proof}

\begin{rem}~
\begin{itemize}
	\item[(i)] Under the assumption  \eqref{eq:confblockHom}, the theorem implies
	\[\Hom_{\mathcal{C}^k}\left(D^{\pm,k,s}_i\fuse D^{\mp,k,-t}_j,\one\right)\cong\delta_{i,j}\delta_{s,t}\C.\]
	\item[(ii)] With present knowledge there is no way to determine $\Hom_{\mathcal{C}^k}\left(\one,D^{\pm,k,s}_i\fuse D^{\mp,k,-t}_j\right)$, but considering Proposition \ref{prop:sl2nonsemi} it is hard to imagine any other result than $\{0\}$ independently of $i,j,s,t$. It then follows that $\mathcal{C}^k$ is nonsemisimple, and furthermore cannot possess a rigid structure.
\end{itemize}
\end{rem}

%%%%%%%%%%%%%%%%%%%%%%%%%%%%%%%%%%%%%%%%%%%%%%%%%%%%%%%%%%%%
\section{Extended chiral symmetry and simple currents}\label{sec:extchir}
As stated in the introduction, it is natural to expect invertible objects, simple currents, in the chiral category, such that $\mathrm{Pic}(\mathcal{C}^k)$ contains at least an infinite cyclic subgroup corresponding to the center of $SL(2,\R)$. 
Comparing with rational WZW models it is furthermore natural to expect that the simple objects $D^{k,s}_0$ are invertible in $\mathcal{C}^k$, satisfying $D^{k,s}_0\fuse V^k\cong V^{k,s}$ for any object $V$. This would then lead to an infinite cyclic subgroup in $\mathrm{Pic}(\mathcal{C}^k)$ isomorphic to $\langle\vartheta_1\rangle$. Note that, of course, this implies that $D^k_0$ itself has trivial fusion with any object and therefore takes the role, as expected, of the tensor unit $\one$.\\

Let us first find gather evidence that $D^k_0$ is a tensor unit in $\mathcal{C}^k$. The strongest indication of this property would be that $\Hom_{\mathcal{C}^k}(D^k_0\fuse W,\one)\cong\Hom_{\mathcal{C}^k}(W,\one)$ for any object $W$. Thus the following theorem.
\begin{thm}\label{thm:unit}
	For any $n$ tuple $V^k_{\vec{\lambda}}$ of $\ens$--modules at level $k$ and any $n+1$ points $(q,\vec{p})$ on $\C\mathbb{P}^1$,
	\[\mathcal{H}\left((q,\vec{p}),(D^k_0,V^k_{\vec{\lambda}})\right)\cong\mathcal{H}(\vec{p},V^k_{\vec{\lambda}}).\]
\end{thm}
The proof of this theorem actually follows the proof in the rational case. To this end one may use a variation of Proposition 2.3 in \cite{Be} (see also \cite{BK}):
\begin{lem}(Beauville)\label{lem:Beauville}
	Let $\mathfrak{g}$ be a simple Lie algebra, $V^k_\lambda$ a $\mathcal{U}(\widehat{\mathfrak{g}})$--module of level $k$,  and $V_\lambda$ a $\mathcal{U}(\mathfrak{g})$--module such that $V^k_\lambda=\mathcal{U}(\widehat{\mathfrak{g}}^-)V_\lambda$ (where $\widehat{\mathfrak{g}}^+$ annihilates $V_\lambda$). Furthermore, let $V^k_{\vec{\mu}}=V_{\mu_1}\oti V_{\mu_2}\oti\cdots\oti V_{\mu_n}$ be a tensor product of arbitrary $\mathcal{U}(\widehat{\mathfrak{g}})$--modules, and let $q$, and $\vec{p}=(p_1,\ldots,p_n)$ denote $n+1$ points on a smooth complex curve $C$. Then
	\begin{equation}
		\left(V^k_\lambda\oti V^k_{\vec{\mu}}\right)^*_{\mathfrak{g}(C-q-\vec{p})}\cong\left(V_\lambda\oti V^k_{\vec{\mu}}\right)^*_{\mathfrak{g}(C-\vec{p})},
	\end{equation}
	where, on the right hand side, $\mathfrak{g}(C-\vec{p})$ acts on $V_\lambda$ by evaluation at $q$.
\end{lem}
\begin{proof}
Also the proof follows \cite{Be}, although with some modifications. By the Riemann--Roch theorem there exists a function $z$ on $C$, vanishing at $q$, such that any $f\in\mathcal{O}(C-q-\vec{p})$ can be written uniquely as $f=\tilde{f}+\sum_{i=1}^Nf_iz^{-1}$, where $\tilde{f}\in\mathcal{O}(C-\vec{p})$, $f_i\in\mathcal{O}_C$, and it follows that $\mathcal{O}(C-q-\vec{p})\cong\mathcal{O}(C-\vec{p})\oplus_{i\in\N}\C z^{-i}$. This implies $\mathfrak{g}(C-q-\vec{p})\cong\mathfrak{g}(C-\vec{p})\oplus\widehat{\mathfrak{g}}^-$. Note that under this isomorphism, the first summand acts on $V^k_\lambda$ by evaulation, while the second summand acts on $V^k_\lambda$ in the natural way and on $V^k_{\vec{\mu}}$ by evaluation.

Let $i^*:\left(V^k_\lambda\oti V^k_{\vec{\mu}}\right)^*\twoheadrightarrow\left(V_\lambda\oti V^k_{\vec{\mu}}\right)^*$ be the surjection transpose to the inclusion $i:V_\lambda\oti V^k_{\vec{\mu}}\hookrightarrow V^k_\lambda\oti V^k_{\vec{\mu}}$.
Then, the restriction $p$ of $i^*$ to the subspace of $\widehat{\mathfrak{g}}^-$--invariants is a linear isomorphism:
\[p:\left(V^k_\lambda\oti V^k_{\vec{\mu}}\right)^*_{\widehat{\mathfrak{g}}^-}\stackrel{\sim}{\rightarrow}\left(V_\lambda\oti V^k_{\vec{\mu}}\right)^*.\]
To see this, pick an arbitrary $\beta\in\left(V^k_\lambda\oti V^k_{\vec{\mu}}\right)^*_{\widehat{\mathfrak{g}}^-}$ and note that, for any $u\oti v\in V^k_\lambda\oti V^k_{\vec{\mu}}$,  the $\widehat{\mathfrak{g}}^-$--invariance implies $\beta(u\oti v)=\beta(u_0\oti v')$ for some $u_0\oti v'\in V_\lambda\oti V^k_{\vec{\mu}}$. It immediately follows that $\mathrm{Ker}(p)=\{0\}$, and $p$ is surjective since any element of $\left(V_\lambda\oti V^k_{\vec{\mu}}\right)^*$ extends uniquely to $\left(V^k_\lambda\oti V^k_{\vec{\mu}}\right)^*_{\widehat{\mathfrak{g}}^-}$ by $\widehat{\mathfrak{g}}^-$--invariance.

The final step of the proof amounts to showing that $p$ is $\mathfrak{g}(C-\vec{p})$--equivariant, which follows if the inclusion $i$ is equivariant. Note that $X\oti f\in\mathfrak{g}(C-\vec{p})$ acts on $V^k_\lambda\oti V^k_{\vec{\mu}}$ as
\begin{eqnarray}
	\big(X\oti f\big)(u\oti v) & = & \Big(\big(X\oti f\big) u\Big)\oti v+ u\oti\Big(\big(X\oti f\big) v\Big)\nonumber\\
	& = & \big(f(q)Xu+ X^+_fu\big)\oti v+u\oti\Big(\big(X\oti f\big) v\Big),\nonumber	
\end{eqnarray}
for some $X^+_f\in\widehat{\mathfrak{g}}^+$. Since $\widehat{\mathfrak{g}}^+$ annihilates $V_\lambda$, it follows that $i\circ (X\oti f)=(X\oti f)\circ i$, so $i$ is $\mathfrak{g}(C-\vec{p})$--equivariant. The transpose of $i$, and in particular its restriction $p$, is then automatically equivariant. Since the isomorphism $p$ is equivariant it restricts further to an isomorphism of $\mathfrak{g}(C-\vec{p})$--invariant subspaces, which completes the proof.
\end{proof}
The proof of Theorem \ref{thm:unit} is now trivial.
\begin{proof}[Proof of Theorem \ref{thm:unit}] According to Lemma \ref{lem:Beauville}, \[\mathcal{H}((q,\vec{p}),(D^k_0,V^k_{\vec{\lambda}}))\cong\left(D_0\oti V^k_{\lambda_1}\oti\cdots\oti V^k_{\lambda_n}\right)^*_{\mathfrak{sl}_2(\C\mathbb{P}^1-\vec{p})},\] but since $D_0\cong\C$ with the trivial action, the latter space is isomorphic to 
\[\left(V^k_{\lambda_1}\oti\cdots\oti V^k_{\lambda_n}\right)^*_{\mathfrak{sl}_2(\C\mathbb{P}^1-\vec{p})} \cong \mathcal{H}(\vec{p},V_{\vec{\lambda}}).\]
\end{proof}
\begin{rem}~
	\begin{itemize}
		\item[(i)] Note that Theorem \ref{thm:unit} requires $n>0$, however it is also trivial to show explicitly that, on $\C\mathbb{P}^1$, $\mathcal{H}(p,D^{\pm,k,s}_j)=\{0\}$ (for $j$ in the unitary range) and $\mathcal{H}(p,D^{k,s}_0)\cong \delta_{s,0}\C$, thus implying $\Hom_{\mathcal{C}^k}(D^{\pm,k,s}_j,\one)=\{0\}$, $\Hom_{\mathcal{C}^k}(D^{k,s}_0,\one)\cong\delta_{s,0}\C$.
		\item[(ii)] In a semisimple rigid category Theorem \ref{thm:unit} implies that $D^k_0$ is isomorphic to the tensor unit, and $\Hom_{\mathcal{C}^k}(D^k_0,\one)\cong\C$ then implies that $\one$ is absolutely simple. Since the category $\mathcal{C}^k$ cannot be assumed semisimple or rigid, however, our results do not allow us to draw these conclusions.
	\end{itemize}
\end{rem}
We would have a strong indication that the modules $D^{k,s}_0$ are invertible with fusion $D^{k,s}_0\fuse V^k\cong V^{k,s}$ if it could be established that $\Hom_{\mathcal{C}^k}(D^{k,s}_0\fuse V^k\fuse W,\one)\cong\Hom_{\mathcal{C}^k}(V^{k,s}\fuse W^s,\one)$ for any simple object $V^k$ and any object $W$. What we have been able to show in this direction is the following theorem.
\begin{thm}\label{thm:sfconfblocks}
	Let $V^k_{\vec{\lambda}}$ be a, possibly empty, $n$ tuple of prolongation modules of $\ena$, and let $X$ and $Y$ be any $\ena$--modules. Then, on $\C\mathbb{P}^1$,
	\begin{equation}
		\mathcal{H}\left((0,\infty,\vec{p}),(X,Y,V^k_{\vec{\lambda}})\right)\cong\mathcal{H}\left((0,\infty,\vec{p}),(X^s,Y^{-s},V^k_{\vec{\lambda}})\right),
	\end{equation}
	for any $s\in\Z$.
\end{thm}
\begin{proof}
Use first Lemma \ref{lem:Beauville} repeatedly to get an isomorphism
\[\left(X\oti Y\oti V^k_{\lambda_1}\oti\cdots\oti V^k_{\lambda_n}\right)^*_{\mathfrak{sl}_2(\C\mathbb{P}^1-0-\infty-\vec{p})}\cong\left(X\oti Y\oti V_{\lambda_1}\oti\cdots\oti V_{\lambda_n}\right)^*_{\mathfrak{sl}_2(\C\mathbb{P}^1-0-\infty)}.\]
Fix the standard local coordinates $z$ around $0$ and $w=1/z$ around $\infty$, it then follows that $\mathfrak{sl}_2(\C\mathbb{P}^1-0-\infty)\cong\mathfrak{sl}_2[z,z^{-1}]$. If $z_i$ is the value of $z$ at the point $p_i$, the coinvariance condition reads:
\[\beta\circ\left(J^a_m\oti\mathbf{1}+\mathbf{1}\oti J^a_{-m}\oti\mathbf{1}+ \mathbf{1}\oti z_1^mJ^a\oti\mathbf{1}+\ldots +\mathbf{1}\oti z_n^mJ^a\right)=0,\]
for $a=0,\pm$ and all $m\in\Z$.
Since, as vector spaces,
\[X^s\oti Y^{-s}\oti V_{\lambda_1}\oti\cdots\oti V_{\lambda_n} = X\oti Y\oti V_{\lambda_1}\oti\cdots\oti V_{\lambda_n}, \]
the coinvariance condition on the dual of the left hand side is expressed using the representation morphism of the right hand side as invariance under the action of
\begin{align}
	(J^0_m-sk\delta_{m,0})\oti\mathbf{1}+\mathbf{1}\oti (J^0_m+sk\delta_{m,0})\oti\mathbf{1}+\mathbf{1}\oti z_1^mJ^0\oti\mathbf{1}+\ldots +\mathbf{1}\oti z_n^mJ^0\nonumber\\
	J^\pm_{m\mp s}\oti\mathbf{1}+\mathbf{1}\oti J^\pm_{-(m\mp s)}\oti\mathbf{1} + \mathbf{1}\oti z_1^mJ^\pm\oti\mathbf{1}+\ldots +\mathbf{1}\oti z_n^mJ^\pm.\nonumber
\end{align}
Were it not for the evaluation terms in the second line, the equivalence of $(X^s\oti Y^{-s}\oti V_{\lambda_1}\oti\cdots\oti V_{\lambda_n})^*_{\mathfrak{sl}_2(\C\mathbb{P}^1-0-\infty)}$ and $(X\oti Y\oti V_{\lambda_1}\oti\cdots\oti V_{\lambda_n})^*_{\mathfrak{sl}_2(\C\mathbb{P}^1-0-\infty)}$ would be immediate. The mismatch due to these terms is, however, easily cured. Note that the map $J^0\mapsto J^0$, $J^\pm\mapsto z^{\pm 1}J^{\pm}$ is an inner automorphism of $\mathfrak{sl}_2$ for any $z\in\C^\times$. Thus, replacing each $V_{\lambda_i}$ with the one obtained by  precomposing the representation morphism with the inner automorphism given by $z_i^{-s}$ gives precisely the required coinvariance condition. Since these modules are equivalent to the $V_{\lambda_i}$, the proof is complete.
\end{proof}
Theorems \ref{thm:unit} and \ref{thm:sfconfblocks} immediately imply the following corollary.
\begin{cor}
	\[\mathcal{H}\left( (0,\infty,\vec{p}),(D^{k,s}_0,V^k,V^k_{\vec{\lambda}})\right)\cong\mathcal{H}\left((\infty,\vec{p}),(V^{k,s},V^k_{\vec{\lambda}})\right),\]
	for any $\ena$--module $V^k$.
\end{cor}
	
\begin{rem}
	Again, under assumption \eqref{eq:confblockHom} the results above imply
		\[\Hom_{\mathcal{C}^k}\left(D^{k,s}_0\fuse V^k\fuse W,\one\right)\cong\Hom_{\mathcal{C}^k}\left(V^{k,s}\fuse W,\one\right)\] for any object $V^k$ and any object $W$ of the form $V^k_{\lambda_1}\fuse\cdots\fuse V^k_{\lambda_n}$ where each $V^k_{\lambda_i}$ is a prolongation module. The difficulty in showing the same result for an arbitrary type of module $W$ appears to be a technical one, and we expect that the stronger result indeed holds. This is, as already mentioned, with the knowledge at hand as strong an indication of $D^{k,s}_0$ being invertible as can possibly obtained.
\end{rem}

~\\
Assume, motivated by the corollary above, that $J\defas D^{k,1}_0$ is an invertible object with infinite order in $\mathcal{C}^k$ (it's inverse being $D^{k,-1}_0$), satisfying
\begin{align}
	J^{\fuse s}\fuse D^{\pm,k,s'}_j &\cong D^{\pm,k,s+s'}_j.
\end{align}
Abusing the term ``primary field'' somewhat (the modules $D^{k,s}_0$ are not highest or lowest weight \ena--modules for $s\neq 0,\pm 1$), a primary field corresponding to the invertible object $J^{\fuse s}$ will have conformal weight
\begin{equation}
	\Delta_{J^s} = \frac{s^2k}{4}.
\end{equation}
It follows that for $s\in 2\Z$ the invertible object $J^{\fuse s}$ is an integer spin simple current, and the collection of even $s$ invertible objects correspond to an infinite cyclic subgroup $H=\langle[J^{\fuse 2}]\rangle\subset\langle [J]\rangle\subset\mathrm{Pic}(\mathcal{C}^k)$. The $H$ orbit of $[D^{\pm,k}_j]\in K_0(\mathcal{C}^k)$ is obviously
\begin{equation}
	H[D^{\pm,k}_j] = \sum_{s\in 2\Z} [D^{\pm,k,s}_j] = [X^{\pm,k}_j],
\end{equation}
where $X^{\pm,k}_j$ is the brick from \eqref{eq:brick2}. Consequently, the brick decomposition of the bulk spectra in \cite{HeHwRS} correspond in part to an integer spin simple current extension. A natural question is whether also the full brick $X^k_j=X^+_j+X^-_j$ is the result of a simple current extension. One possibility is that there exists an invertible object $L\in\mathrm{Ob}(\mathcal{C}^k)$ with the property
\begin{equation}
	L\fuse D^{\pm,k,s}_j\cong D^{\mp,k,s+2l}_j,
\end{equation}
for some $l\in\Z$. We have not been able to identify a candidate for such an invertible object. Note that the action of $J$ itself would shift $j$ to $k/2-j$, and so does not qualify.

%%%%%%%%%%%%%%%%%%%%%%%%%%%%%%%%%%%%%%%%%%
\section{Discussion}\label{sec:disc}
The aim of this paper was to investigate the structure of the chiral category $\mathcal{C}^k$ underlying an operator formulation of the $SL(2,\R)$ WZW model at negative integer level $k$. To this end we have studied properties of a certain class of \ena--modules obtained by prolonging the discrete series of weight modules and the trivial module of \ens, and by acting on these modules with spectral flow automorphisms.
Briefly, the main conclusions are:
\begin{itemize}
	\item There is a natural notion of contragredient \ena--module which corresponds to conjugate primary fields.
	\item There are strong indications that there exists a nonzero morphism from $D^{\pm,k,s}_j\fuse(D^{\pm.k.s}_j)^\cont$ to the tensor unit $\one$, in contrast to what the unitary Clebsch--Gordan series of $sl(2,\R)$ would suggest.
	\item There are strong indications that there are \emph{no} nonzero morphisms from $\one$ to $D^{\pm,k,s}_j\fuse(D^{\pm.k.s}_j)^\cont$, thus implying that the chiral category $\mathcal{C}^k$ contains nonsimple but indecomposable objects, and that $\mathcal{C}^k$ canot be rigid.
	\item The simple objects $D^{k,s}_0$ obtained by acting with the spectral flow automorphisms on the vacuum module satisfy consistency requirements of being invertible objects with fusion
\[D^{k,s}_0\fuse D^{k,s'}_0\cong D^{k,s+s'}_0,\quad D^{k,s}_0\fuse  D^{\pm,k,s'}_j\cong D^{\pm,k,s+s'}_j,\]
and in particular, $D^k_0$ satisfies properties required from the tensor unit of $\mathcal{C}^k$.
Consequently we have found evidence that the bricks $X^{\pm,k}_j$ occurring in proposed bulk spectra on the single cover of $SL(2,\R)$ are obtained by an integer spin simple current extension.
\end{itemize}

It is worth emphasizing that the representation theory of $sl(2,\R)$ is very well understood, so the results presented in section \ref{sec:sl2discr} 
should likely be possible to extract from the literature. Only very basic tools are needed, however, making it easier to derive the results directly. The results displaying a significant difference between unitary and unitarizable modules are particularly interesting, at least for applications to the $SL(2,\R)$ WZW model, and deserve to be investigated further.
As an example it would be interesting to examine the full \emph{monoidal}  subcategory of $\mathrm{Rep}\left((\mathfrak{sl}_2,\widetilde{\omega})\right)$ generated from simple unitarizable weight modules. Results in this direction would provide nice hints to the monoidal structure of $\mathcal{C}^k$.\\

Many of the results in section \ref{sec:proldiscr} appear to be more or less well known in the community, as evidenced by statements in the literature. It would, however, be wrong to claim that the structure of these modules is completely understood, as indicated by the result in Proposition \ref{prop:scquotient}. As an example, a better understanding of  the decomposition of this class of modules in terms of various $\mathfrak{sl}_2$ subalgebras is likely to be useful in, for instance, an attempt to generalize the Kazhdan--Lusztig \cite{KaLu} construction of the fusion product to these modules.

It would in addition be useful to understand better the construction of these modules using free fields. Free field realizations have been employed in applications to $SL(2,\R)$ models for a number of years \cite{Giveon}, and even though the understanding of these have increased over time (see e.g. the recent \cite{Nun3,Nun4}), the understanding of the free field construction of the relevant modules is still incomplete. Since the modules discussed here are simple and Verma, or straightforwardly obtained from such modules, one would guess that the BRST construction of Bernard \& Felder \cite{BeFe} will degenerate or trivialize in the $SL(2,\R)$ case, leaving the role of screening currents less clear. Note, however, that Proposition \ref{prop:scquotient} implies that the modules $D^{k,\pm 1}_0$ are proper quotients of Verma modules. A construction of these using free fields will therefore involve a non trivial BRST complex.\\

The results of section \ref{sec:cont} are the main results of the present article, and have important consequences for the $SL(2,\R)$ WZW model. For instance, if the chiral category is nonsemisimple and nonrigid great care is called for in interpreting alleged fusion rules. As an elementary illustrative example, let $V^k_i$, $i\in\mathcal{I}$ denote a set of simple \ena--modules, and consider the two sets of quantities (``fusion rule coefficients'') $\mathcal{N}_{i,j}^{\phantom{i,j}l}$ and  $\widetilde{\mathcal{N}}_{i,j}^{\phantom{i,j}l}$ defined by
\begin{eqnarray*}
	\mathcal{N}_{i,j}^{\phantom{i,j}l} &\defas & \mathrm{dim}_\C\Hom_{\mathcal{C}^k}(V^k_i\fuse V^k_j,V^k_l)\\
	\widetilde{\mathcal{N}}_{i,j}^{\phantom{i,j}l} &\defas & \mathrm{dim}_\C\Hom_{\mathcal{C}^k}(V^k_l,V^k_i\fuse V^k_j).
\end{eqnarray*}
If $V^k_i\fuse V^k_j$ is isomorphic to a direct sum of simple objects, such a decomposition immediately implies $\mathcal{N}_{i,j}^{\phantom{i,j}l} = \widetilde{\mathcal{N}}_{i,j}^{\phantom{i,j}l}$. If $\mathcal{C}^k$ is rigid the duality implies $\mathcal{N}_{i,j}^{\phantom{i,j}l} = \widetilde{\mathcal{N}}_{\bar{\jmath},\bar{\imath}}^{\phantom{i,j}\bar{l}}$,  which furthermore coincides with
\begin{equation*}
	\mathrm{dim}_\C\Hom_{\mathcal{C}^k}(V^k_i\fuse V^k_j\fuse V^k_{\bar{l}},\one).
\end{equation*}
If, on the other hand, $\mathcal{C}^k$ is neither semisimple nor rigid, none of these relations hold in general.  These types of remarks will have crucial implications for factorization properties of conformal blocks.
 It is therefore of great importance to put these results on a more solid foundation. 
In order to achieve this it would be useful to have an independent construction of the category $\mathcal{C}^k$. In this direction it is worth investigating whether the modules discussed in here also appear as (perhaps generalized) modules of a conformal vertex algebra related to the untwisted affine Lie algebra $\widehat{\mathfrak{sl}}_2$ at negative integer level. Since there is by now a well developed theory of strongly graded (generalized) modules of (strongly graded) conformal vertex algebras \cite{HLZ}, such a construction would likely imply a number of immediate results. Another attractive feature of this approach, if successful, is the results concerning semi--rigidity of modules of certain vertex operator algebras \cite{Miyamoto}.
A  less direct approach would be to try and find a Kazhdan--Lusztig dual Hopf algebra (or, perhaps in this case, bialgebra), governing the structure of the category $\mathcal{C}^k$. Such algebras have recently been studied in some examples of nonrational CFT's, a class of logarithmic models \cite{FGST}, with interesting results. The conventional method to obtain a quantum group from a CFT goes via a free field realization and screening currents, again stressing the need for a better understanding of these in the $SL(2,\R)$ case.\\

Concerning the issue of invertible objects discussed in section \ref{sec:extchir}, a strengthening of Theorem \ref{thm:sfconfblocks} would be comforting, and, we believe, possible. Again, an alternative to simply strengthening the theorem would be to evaluate the corresponding $n$--point functions in a free field realization. 
The nature of a hypothetical extension, beyond the simple current extension discussed above, of the chiral algebra with simple modules corresponding to the bricks $X^k_j$, remains elusive. In the world of quantum groups one may consider quantized universal enveloping algebras of the Lie algebra $sl(2,\R)$. There is an extension of these algebras due to Faddeev \cite{Faddeev}, claimed to be relevant for Liouville theory, known as the modular double of $\mathcal{U}_q(sl(2,\R))$. Due to the relation between the $SL(2,\R)$ WZW model and Liouville theory one might hope that this extension is in some way related to the extension at hand. There seems, however, to be no independent reason to believe that this is the case.

%%%%%%%%%%%%%%%%%%%%%%%%%%%%%%%%%%%%%%%%%
\appendix

\section{Continuation of the proof of Theorem \ref{thm:2ptblocks}}\label{sec:contproof}
\begin{proof} ($s\neq t$) According to Theorem \ref{thm:sfconfblocks} it is enough to consider the case where $t=0$ and $s\neq 0$, and by Lemma \ref{lem:Beauville} there is an isomorphism
\[\left(D^{+,k,s}_i\oti D^{-,k}_j\right)^*_{\mathfrak{sl}_2(\C\mathbb{P}^1-0-\infty)}\cong\left(D^{+,k,s}_i\oti D^-_j\right)^*_{\mathfrak{sl}_2(\C\mathbb{P}^1-0)}.\]
The coinvariance condition on the right hand side amounts to requiring
\begin{align}
	\beta\circ(J^a_{-n}\oti\mathbf{1}) &= 0,\quad a=0,\pm,\ n\in\N\label{eq:ci1}\\
	\beta\circ(J^a_0\oti\mathbf{1}+\mathbf{1}\oti J^a) &= 0,\quad a=0,\pm.\label{eq:ci2}
\end{align}
Applying \eqref{eq:ci2} it follows that for any $u\oti |j,m\rangle$ such that $\beta(u\oti|j,m\rangle)\neq 0$, the latter expresion is proportional to $\beta(u'\oti|j,0\rangle)$ for some $u'$. The condition \eqref{eq:ci2} with $a=0$ then gives $2j=2i+sk-2p$ for some $p\in\Z$. The conditions \eqref{eq:ci1} imply that $u'$ must be of Virasoro level $\leq 0$. A comparison with \eqref{wts1} (or the weight diagram for $D^{+,k,s}_i$) then reveals that for $s\geq 1$, $(2i+sk-2p)\leq (2i+sk)<sk<0$, while for $s\leq -1$, $(2i+sk-2p)\geq (2i+sk)>0$. Since $k<2j<0$ the condition $2j=2i+sk-2p$ cannot be satisfied either for $s>0$ or for $s<0$. It follows that the space of coinvariants is $0$ whenever $s\neq t$, finishing the proof of Theorem \ref{thm:2ptblocks}.
\end{proof}

\paragraph{Acknowledgements} I am grateful to J\"urgen Fuchs for comments. This work is supported in parts by NSFC grant No.~10775067  as well as   Research Links Programme of Swedish Research Council under contract No.~348-2008-6049.

\end{document}